\documentclass[11pt]{article}
\usepackage{amsmath,amssymb}
\usepackage{amsthm}
\usepackage{titlesec}
\usepackage{indentfirst}

\usepackage{bbm}
\usepackage{mathrsfs}
\usepackage{amsfonts}

\usepackage{bm}
\usepackage{xy}      
\xyoption{all}
\usepackage{graphicx}

\theoremstyle{plain}
\newtheorem{thm}{Theorem}
\newtheorem{Pp}[thm]{Proposition}
\newtheorem{Co}[thm]{Corollary}
\newtheorem{Lm}[thm]{Lemma}
\newtheorem{Df}[thm]{Definition}
\newtheorem{Rm}[thm]{Remark}

\renewcommand\thesection{\arabic{section}.\kern -.3em}
\renewcommand{\thesubsection}{\arabic{section}.\arabic{subsection}.\kern -.5em}

\newcommand{\Prob} {{\textrm{Prob}}}

\begin{document}
\title{{\bf Exponentially more concise quantum recognition of non-RMM regular languages}}
\author{Daowen Qiu$^{a,b,}$\thanks{Corresponding author.
{\it E-mail address:} issqdw@mail.sysu.edu.cn (D. Qiu).} ,\hskip 2mm  Lvzhou Li$^{a}$,
\hskip 2mm Paulo Mateus$^{b,}$\thanks{
{\it E-mail address:} \{pmat,acs\}@math.ist.utl.pt.} ,  \hskip 2mm Amilcar Sernadas$^{b,\ddag}$\\
\small{{\it $^a$Department of
Computer Science, Sun Yat-sen University, Guangzhou 510006, China}}\\
\small{{\it $^b$SQIG--Instituto de Telecomunica\c{c}\~{o}es, Departamento de Matem\'{a}tica, Instituto Superior T\'{e}cnico, }}\\
\small{{\it  Technical University of Lisbon, Av. Rovisco Pais 1049-001, Lisbon, Portugal}}\\}
\date{ }
\maketitle

\begin{abstract}

We show that there are quantum devices that accept all regular languages and that are exponentially more concise than deterministic finite automata (DFA). For this purpose, we introduce a new computing model of {\it one-way quantum finite automata} (1QFA), namely, {\it one-way quantum finite
automata together with classical states} (1QFAC), which extends naturally both measure-only 1QFA and DFA and whose state complexity is upper-bounded by both. The original contributions of the paper are the following. First, we show that the set of languages accepted by 1QFAC with bounded error consists precisely of all regular languages. Second, we prove that 1QFAC are at most exponentially more concise than DFA. Third, we show that the previous bound is tight for families of regular languages that are not recognized by measure-once (RMO), measure-many (RMM) and multi-letter 1QFA. 
Fourth, we give a polynomial-time algorithm for determining whether any two 1QFAC are equivalent. Finally, we show that state minimization of 1QFAC is decidable within EXPSPACE. We conclude the paper by posing some open problems.

\vskip 2mm

\end{abstract}

\section{Introduction}

{\it Quantum finite automata} (QFA) can be thought of as a
theoretical model of quantum computers in which the memory is
finite and described by a finite-dimensional state space
\cite{ABGKMT}, as finite automata are a natural model for
classical computing with finite memory \cite{HU79}. As mentioned
in \cite{Hir08}, one of the motivations to study QFA is to provide
some ideas to investigate the relation of classical and quantum
computational complexity classes. This kind of theoretical models
was firstly studied by Moore and Crutchfield \cite{MC00}, Kondacs
and Watrous \cite{KW97}, and then Ambainis and Freivalds
\cite{AF98}, Brodsky and Pippenger \cite{BP02}, and other authors
(e.g., the references in \cite{QL08}). The study of QFA is mainly divided into
two ways: one is {\it one-way quantum finite automata} (1QFA)
whose tape heads move one cell only to right at each evolution,
and the other {\it two-way quantum finite automata} (2QFA), in
which the tape heads are allowed to move towards right or left, or
to be stationary.  According to the measurement times in a computation,
1QFA have two types: {\it measure-once} 1QFA (MO-1QFA) initiated
by Moore and Crutchfield \cite{MC00} and {\it measure-many} 1QFA
(MM-1QFA) studied first by Kondacs and Watrous \cite{KW97}. In
MO-1QFA, there is only a measurement for computing each input
string, performing after reading the last symbol;  in contrast, in
MM-1QFA, measurement is performed after reading each symbol,
instead of only the last symbol. Notably, QFA have been applied to quantum interactive proof systems \cite{NY09}.

MM-1QFA can accept more languages than MO-1QFA with bounded
error \cite{AF98}, but both of them accept proper subsets of regular
languages \cite{BP02,BMP03}. Another model of 1QFA with a
measurement is called {\it multi-letter} 1QFA, proposed in
\cite{BRS07}. In multi-letter 1QFA, there are multi-reading heads.
Roughly speaking, a $k$-letter 1QFA is not limited to seeing only
one, the just-incoming input letter,  but can see several earlier
received letters as well. Though multi-letter 1QFA can accept some
regular languages not acceptable by MM-1QFA, they still accept a
proper subset of regular languages. Nevertheless, as Ambainis et al \cite{AN09} mentioned, sufficient general 1QFA can indeed accept the same set of languages as DFA,
for example, 1QFA {\it with control languages} (1QFACL, for short)
proposed in \cite{BMP03} accept all regular languages (and only
regular languages) \cite{BMP03,MP06}. However, the measurements in
1QFACL differ from those in MM-1QFA proposed in \cite{KW97}.

Paschen \cite{Pas00} presented a different 1QFA by adding some ancilla qubits to
avoid the restriction of unitarity, and this model is called an ancilla 1QFA. Indeed, in ancilla 1QFA, the transition function  corresponding to every input symbol is described by  an isometry  mapping, instead of a unitary operator. In \cite{Pas00}, it was proved that ancilla
1QFA can recognize any regular language with certainty.  With the idea in Bennett \cite{Ben73},
Ciamarra  \cite{Cia01} proposed another model of 1QFA
 whose computational power was shown to be  at least equal to
that of classical automata. For convenience,  we call the 1QFA
defined in \cite{Cia01} as {\it Ciamarra  1QFA} named after the
author. In fact, the internal state of a Ciamarra  1QFA evolves by a
trace-preserving quantum operation. In addition, in \cite{LQZLW09} it was proved that both ancilla 1QFA and Ciamarra  1QFA recognize only regular languages. Recently, it was proved that MO-1QFA and MM-1QFA with mixed states and trace-preserving quantum operations, instead of unitary operators, as the evolutions of states, can accept all and only regular languages \cite{LQZLW09}

These 1QFA indicated above can accept all regular languages, but their architectures are much more complicated than MO-1QFA, and more difficult to be implemented physically with present technology.
Hence, proposing and exploring practical models of
quantum computation is an important research problem and provide relevant insights
to study physical models of quantum computers. Indeed,
 motivated by the implementations of quantum
computers using nucleo-magnetic resonance (NMR), Ambainis et al. \cite{ABGKMT}
proposed another model of 1QFA, namely, {\em Latvian} 1QFA (L-1QFA, for short). In
L-1QFA, measurement is also allowed after reading each input
symbol, but they accept a proper subset of regular languages
\cite{ABGKMT}. Notably,  the languages recognized with unbounded error by
QFA have been discussed in \cite{YC08,YC10}.

Though ancilla 1QFA and Ciamarra  1QFA can accept all regular languages, their evolution operators of states are general quantum operations instead of unitary operators.
1QFA with pure states and unitary evolutions usually have less recognition power than
{\em deterministic finite automata} (DFA) due to the unitarity (reversibility) of
quantum physics and the finite memory of finite automata. 1QFACL can accept all regular languages but their measurement is quite complicated.
However, one would expect a quantum variant to exceed (or at least to be not weaker than) the corresponding classical computing model, and such quantum computing models are practical and feasible as well. For this reason, we think that a quantum computer should inherit the characteristics of classical computers but further advance classical component by employing quantum mechanics principle.

Motivated by this idea, we propose a new model of quantum automata including a classical component, i.e., we reformulate the definition of this new model of MO-1QFA,
 namely, 1QFA {\em together with classical states} (1QFAC, for
short), and in particular, we investigate some of the basic properties of this new model. As MO-1QFA \cite{MC00,BP02}, 1QFAC execute only a
measurement for computing each input string, following the last
symbol scanned.  In this new model, we preserve the component of
DFA that is used to control
the choice of unitary transformation for scanning each input
symbol. We now describe roughly a 1QFAC ${\cal A}$ computing an input
string, delaying the details until Section 2.

At start up, automaton ${\cal A}$ is in an initial classical state and in an initial
quantum state. By reading the first input symbol, the
classical transformation results in a new classical state as
current state, and,  the initial classical state together with current input symbol assigns a
unitary transformation to process the initial quantum state, leading to a new quantum state as current state.
Afterwards, the machine reads the next input symbol, and similar
to the above process, its classical state will be updated by
reading the current input symbol and, at the same time, with the current classical
state and input symbol, a new unitary transformation is assigned to
execute the current quantum state. Subsequently, it continues to
operate for the next step, until the last input symbol has been
scanned. According to the last classical state, a measurement is assigned to perform on the final quantum
state, producing a result of accepting or rejecting the input
string.

Therefore, a 1QFAC performs only one measurement for
computing each input string, doing so after reading the last symbol. However, the measurement is chosen according to the last classical state reached after scanning the input string.
Thus, when a 1QFAC has only one classical state, it reduces to an MO-1QFA \cite{MC00,BP02}. On the one hand, 1QFAC model develops MO-1QFA by adding DFA's component, and on the other hand, 1QFAC advance DFA by employing the fundamentals of quantum mechanics.

We want to stress that 1QFAC are not the one-way version of  {\it two-way finite automata with quantum and
classical states} (2QCFA for short) proposed by Ambains and Watrous \cite{AW02}, and this version has been preliminarily  considered in \cite{ZQL12}. One
of the differences is that, according to the definition of 2QCFA \cite{AW02},
in the one-way version of 2QCFA, after
the tape head reads an input symbol, either a measurement or a
unitary transformation is performed, while in 1QFAC there is no
intermediate measurement, and a single measurement is performed only after scanning the input string.

Though 1QFAC make only one measurement for computing each input
string and the evolutions of states are unitary instead of general operations,  the set of languages accepted by 1QFAC (with no error)
consists precisely of all regular languages. As we know, the set of
languages accepted by 1QFACL is constituted by all regular
languages \cite{MP06}, but 1QFACL need measurement after reading
each input symbol and the measurement is not only restricted to
accepting, rejecting, and non-halting, but also other results
related to the control language attached to the machine.
Therefore, the computing process of a 1QFACL is usually much more
complicated than that of a 1QFAC. On the other hand, measuring may
lead to more errors for the machine.

Since 1QFA do not have more power than DFA in terms of accepting languages, it is more important to discover the space-efficiency of 1QFA compared with other one-way automata. The first  important result is by Ambainis and Freivalds \cite{AF98}, who proved that MO-1QFA need exponentially less number of states than DFA for accepting some languages. (Recently, Ambainis and Nahimovs \cite{AN09} have further improved this result.)

In \cite{BMP05,MP02, MP06,MP07}, Bertoni, Mereghetti, and Palano further proved that MO-1QFA have space-efficient advantage over DFA for accepting some languages. As mentioned before, MO-1QFA with mixed states and general quantum operations can accept all regular languages. Indeed, Freivalds et al. \cite{FOM09} proved that MO-1QFA with mixed states are super-exponentially more concise than MO-1QFA with pure states.

It should be stressed that 1QFAC can accept regular languages
with exponentially less states than the corresponding DFA
\cite{HU79}, and for which there is no MO-1QFA \cite{MC00}, nor MM-1QFA
\cite{KW97}, nor multi-letter 1QFA \cite {BRS07} that can accept
them with bounded error. Hence, in a way, 1QFAC can be
thought of as a more practical model of QFA, showing better state
complexity than DFA due to its quantum computing component, and
stronger recognition power of languages than MO-1QFA, MM-1QFA, and
multi-letter 1QFA. Furthermore, for accepting the same regular
language, we will show that 1QFAC have better state complexity
than 1QFACL.

\subsection*{Original contributions}

The main technical contributions of the paper contain five
aspects. In Section 2, after reviewing some existing 1QFA models,
we  define 1QFAC formally and then prove that the set of languages
accepted by 1QFAC with bounded error includes all
regular languages.

Then, in Section 3, we study the state complexity of 1QFAC. We prove that if $L$ is accepted by a 1QFAC ${\cal M}$ with bounded error, then $L$ is regular and $kn=\Omega( log ~m)$ where $k$ and $n$ denote  numbers of classical states and quantum basis states of ${\cal M}$, respectively, and $m$ is the state number of the minimal DFA accepting $L$. Then, we verify this bound is indeed tight, since we further show that, for any  prime $m\geq 2$, there exists a regular language $L_m$ whose minimal DFA needs $m+1$ states and neither MO-1QFA nor MM-1QFA can accept $L_m$, but there exists a 1QFAC accepting $L_m$ with only two classical states and $O(\log(m))$ quantum basis states. In addition, we show that, for any $m\geq 2$, and any input string $z$, there exists a regular language $L_z(m)$ that can not be accepted by any multi-letter 1QFA or MO-1QFA, but there exists a 1QFAC ${\cal A}_m$ accepting it with only 2 classical states and $O(\log(m))$ quantum basis states. In contrast, the minimal DFA accepting $L_z(m)$ has $(|z|+1)m$ states, where $|z|$ denotes the length of $z$.

In Section 4, we study the equivalence problem of 1QFAC. Any two
1QFAC ${\cal A}_1$ and ${\cal A}_2$ over the same input alphabet
$\Sigma$ are equivalent (resp. $k$-equivalent) iff their
probabilities for accepting any input string (resp. length not
more than $k$) are equal. We reformulate any given 1QFAC with a
bilinear computing machine. According to
\cite{Rab63,Paz71,Tze92}, it follows that 1QFAC ${\cal
A}_1$ and ${\cal A}_2$ over the same input alphabet $\Sigma$ are
equivalent if and only if they are $(k_{1}n_{1})^2+
(k_{2}n_{2})^2-1$-equivalent, and furthermore there exists a polynomial-time
$O([(k_{1}n_{1})^2+ (k_{2}n_{2})^2]^4)$ algorithm for determining
their equivalence, where $k_{1}$ and $k_{2}$ are the
numbers of classical states of ${\cal A}_1$ and ${\cal A}_2$, as well as $n_{1}$ and
$n_{2}$ are the numbers of quantum basis states of ${\cal
A}_{1}$ and ${\cal A}_{2}$, respectively.

Finally, in Section 5, we show that minimization of a 1QFAC is decidable in EXPSPACE. To this end, we capitalize on the results of Section 4 and on Renegard's algorithm  \cite{Ren88} for sampling semialgebraic sets.

In general, notation used in this paper will be explained whenever
new symbols appear.  A language $L$ over alphabet $\Sigma$ is accepted by a computing model with {\em bounded error} if there exist $\lambda>0$ and $\epsilon>0$ such that the accepting probability for $x\in L$ is at least $\lambda+\epsilon$ and the accepting probability for $x\not\in L$ is at most $\lambda-\epsilon$.  In this paper, we always consider the accepting scheme of machines to be bounded error unless we emphasize otherwise. Throughout this paper, the notation $\|. \|$ denotes the Euclid norm of a vector.

\section{One-way quantum finite automata together with classical states}

In this section, we introduce the definition of 1QFAC and then
prove its recognition power of languages.  For the sake of readability, we first recall the definitions of MO-1QFA, MM-1QFA,  multi-letter 1QFA, and 1QFACL.

\subsection{Review of other one-way quantum finite automata}

An MO-1QFA is defined as a quintuple
${\cal A}=(Q, \Sigma, |\psi_{0}\rangle,\{U(\sigma)\}_{\sigma\in\Sigma},Q_{acc})$,
where $Q$ is a set of finite states, $|\psi_{0}\rangle$ is the initial state
that is a superposition of the states in $Q$, $\Sigma$ is a finite
input alphabet, $U(\sigma)$ is a unitary matrix for each
$\sigma\in\Sigma$, and $Q_{acc}\subseteq Q$ is the
set of accepting states.

As usual, we identify $Q$ with an orthonormal base of a complex
Euclidean space and every state $q\in Q$ is identified with a
basis vector, denoted by Dirac symbol $|q\rangle$ (a column
vector), and $\langle q|$ is the conjugate transpose of
$|q\rangle$. We describe the computing process for any given input
string $x=\sigma_{1}\sigma_{2}\cdots\sigma_{m}\in\Sigma^{*}$. At
the beginning the machine ${\cal A}$ is in the initial state
$|\psi_{0}\rangle$, and upon reading $\sigma_{1}$, the
transformation $U(\sigma_{1})$ acts on $|\psi_{0}\rangle$. After
that, $U(\sigma_{1})|\psi_{0}\rangle$ becomes the current state
and the machine reads $\sigma_{2}$. The process continues until
the machine has read $\sigma_{m}$ ending in the state
$|\psi_{x}\rangle=U(\sigma_{m})U(\sigma_{m-1})\cdots
U(\sigma_{1})|\psi_{0}\rangle$. Finally, a measurement is
performed on $|\psi_{x}\rangle$ and the accepting probability
$p_{a}(x)$ is equal to
\[
p_{a}(x)=\langle\psi_{x}|P_{a}|\psi_{x}\rangle=\|P_{a}|\psi_{x}\rangle\|^{2}
\]
where $P_{a}=\sum_{q\in Q_{acc}}|q\rangle\langle q|$ is the
projection onto the subspace spanned by $\{|q\rangle: q\in
Q_{acc}\}$.

Now we further recall the definition of multi-letter QFA
\cite{BRS07}.

A $k$-letter 1QFA ${\cal A}$ is defined as a quintuple ${\cal
A}=(Q,\Sigma, |\psi_{0}\rangle, \nu,Q_{acc})$ where $Q$, $|\psi_{0}\rangle$, $\Sigma$, $Q_{acc}\subseteq Q$, are the same as those in MO-1QFA above, and $\nu$ is a function that assigns a unitary
transition matrix $U_{w}$ on $\mathbb{C}^{|Q|}$ for each string
$w\in (\{\Lambda\}\cup\Sigma)^{k}$, where $|Q|$ is the cardinality
of $Q$.

The computation of a $k$-letter 1QFA ${\cal A}$ works in the same
way as the computation of an MO-1QFA, except that it applies
unitary transformations corresponding not only to the last letter
but the last $k$ letters received. When
$k=1$, it is exactly an MO-1QFA as defined before. According
to \cite{BRS07,QY09}, the languages accepted by $k$-letter 1QFA  are a proper subset of regular languages for any $k$.

An MM-1QFA is defined as a 6-tuple
${\cal A}=(Q,\Sigma,|\psi_{0}\rangle,\{U(\sigma)\}_{\sigma\in\Sigma\cup
\{\$\}},Q_{acc},Q_{rej})$, where $Q,Q_{acc}\subseteq
Q,|\psi_{0}\rangle,\Sigma,\{U(\sigma)\}_{\sigma\in\Sigma\cup
\{\$\}}$ are the same as those in an MO-1QFA defined above,
$Q_{rej}\subseteq Q$ represents the set of rejecting states, and
$\$\not\in\Sigma$ is a tape symbol denoting the right end-mark.
For any input string
$x=\sigma_{1}\sigma_{2}\cdots\sigma_{m}\in\Sigma^{*}$, the
computing process is similar to that of MO-1QFAs except that after
every transition, ${\cal A}$ measures its state with respect to the three
subspaces that are spanned by the three subsets $Q_{acc},
Q_{rej}$, and $Q_{non}$, respectively, where $Q_{non}=Q\setminus
(Q_{acc}\cup Q_{rej})$. In other words, the projection measurement
consists of $\{P_{a},P_{r},P_{n}\}$ where $P_{a}=\sum_{q\in
Q_{acc}}|q\rangle\langle q|$, $P_{r}=\sum_{q\in
Q_{rej}}|q\rangle\langle q|$, $P_{n}=\sum_{q\in Q\setminus
(Q_{acc}\cup Q_{rej})}|q\rangle\langle q|$. The machine stops
after the right end-mark $\$$ has been read. Of course, the
machine may also stop before reading  $\$$ if the current state, after
the machine reading some $\sigma_{i}$ $(1\leq i\leq m)$, does not
contain the states of $Q_{non}$. Since the measurement is
performed after each transition with the states of $Q_{non}$ being
preserved, the accepting probability $p_{a}(x)$ and the rejecting
probability $p_{r}(x)$ are given as follows (for convenience, we
denote $\$=\sigma_{m+1}$):
\[
p_{a}(x)=\sum_{k=1}^{m+1}\|P_{a}U(\sigma_{k})\prod_{i=1}^{k-1}(P_{n}U(\sigma_{i}))|\psi_{0}\rangle\|^{2},
\]
\[
p_{r}(x)=\sum_{k=1}^{m+1}\|P_{r}U(\sigma_{k})\prod_{i=1}^{k-1}(P_{n}U(\sigma_{i}))|\psi_{0}\rangle\|^{2}.
\]
Here we define $\prod_{i=1}^{n}A_i= A_nA_{n-1}\cdots A_1$.

Bertoni {\it et al} \cite{BMP03}
introduced a 1QFA, called 1QFACL that allows a more
general measurement than the previous  models. Similar to the case
in MM-1QFA, the state of this model can be observed at each step,
but an observable ${\cal O}$ is considered with a fixed, but
arbitrary, set of possible results ${\cal C}=\{c_1,\dots,c_n\}$,
without limit to $\{a,r,g\}$ as in MM-1QFA. The accepting
behavior in this model is also different from that of the previous
models. On any given input word $x$, the computation displays a
sequence $y\in {\cal C}^{*}$ of results of ${\cal O}$ with a
certain probability $p(y|x)$, and the computation is accepted if
and only if $y$ belongs to a fixed regular language ${\cal
L}\subseteq {\cal C}^{*}$.   Bertoni {\it et al}
\cite{BMP03} called  such a language ${\cal L}$ {\it control
language}.

More formally, given an input alphabet $\Sigma$ and the end-marker
  symbol $\$\notin\Sigma$, a 1QFACL over the working
  alphabet $\Gamma=\Sigma\cup\{\$\}$ is a five-tuple ${\cal
  M}=(Q, |\psi_{0}\rangle,\{U(\sigma)\}_{\sigma\in\Gamma},{\cal O},{\cal L})$, where
\begin{itemize}

\item $Q$, $|\psi_{0}\rangle$ and $U(\sigma)$ $(\sigma\in\Gamma)$ are defined
as in the case of MM-1QFA;

\item ${\cal O}$ is an observable with the set of possible results
${\cal C}=\{c_1,\dots,c_s\}$ and  the projector set
$\{P(c_i):i=1,\dots,s\}$ of which $P(c_i)$ denotes the projector
onto the eigenspace corresponding to $c_i$;

\item  ${\cal L}\subseteq{\cal C}^{*}$ is a regular language
(control language).

\end{itemize}

The input word $w$ to 1QFACL ${\cal M}$ is in the form:
$w\in\Sigma^{*}\$$, with symbol $\$$ denoting the end of a word.
Now, we define the behavior of ${\cal M}$ on word $x_1\dots
x_n\$$. The computation starts in the state $|\psi_{0}\rangle$, and then the
transformations associated with the symbols in the word  $x_1\dots
x_n\$$ are applied in succession. The transformation associated
with any symbol $\sigma\in\Gamma$ consists of two steps:
\begin{enumerate}
\item[1.] First, $U(\sigma)$ is applied to the current state
$|\phi\rangle$ of ${\cal M}$, yielding the new state
$|\phi^{'}\rangle=U(\sigma)|\phi\rangle$.

\item[2.] Second, the observable ${\cal O}$ is measured on
$|\phi^{'}\rangle$. According to quantum mechanics principle, this
measurement yields result $c_k$ with probability
$p_k=||P(c_k)|\phi^{'}\rangle||^2$, and the state of ${\cal M}$
collapses to $P(c_k)|\phi^{'}\rangle  /\sqrt{p_k}$.
\end{enumerate}

Thus, the computation on word $x_1\dots x_n\$$ leads to a sequence
$y_1\dots y_{n+1}\in {\cal C}^{*}$ with probability $p(y_1\dots
y_{n+1}|x_1\dots x_n\$)$ given by
\begin{equation}
p(y_1\dots y_{n+1}|x_1\dots
x_n\$)=\|  \prod^{n+1}_{i=1}P(y_i)    U(x_i)              |\psi_{0}\rangle\|^2,
\end{equation}
where we let $x_{n+1}=\$$ as stated before. A computation leading
to the word $y\in {\cal C}^{*}$ is said to be  accepted if $y\in
{\cal L}$. Otherwise, it is rejected. Hence, the accepting probability of 1QFACL ${\cal M}$ is
defined as:
\begin{equation}
{\cal P}_{\cal M}(x_1\dots x_n)=\sum_{y_1\dots y_{n+1}\in {\cal
L}}p(y_1\dots y_{n+1}|x_1\dots x_n\$).\label{f_CL}
\end{equation}

\subsection{One-way quantum finite automata together with classical states}

In Section 1, we gave the motivation for introducing the new one-way quantum finite automata model, i.e., 1QFAC. We now define formally the  model.
To this end, we need the following notations. Given a finite set $B$, we denote by ${\cal H}(B)$ the Hilbert space freely generated by $B$. Furthermore, we denote by $I$ and $O$ the identity operator and zero operator on ${\cal H}(Q)$, respectively.

\begin{Df}\em \label{Df1}
A 1QFAC ${\cal A}$ is defined by a 9-tuple
\[
{\cal A}=(S,Q,\Sigma,\Gamma, s_{0},|\psi_{0}\rangle, \delta,
\mathbb{U}, {\cal M})
\]
where:
\begin{itemize}
\item $\Sigma$  is a finite set (the {\it input alphabet});

\item $\Gamma$  is a finite set (the {\it output alphabet});

\item $S$ is a finite set (the set of {\em classical states});

\item $Q$ is a finite set (the {\em quantum state basis});

\item $s_{0}$ is an element of $S$ (the {\em initial classical state});

\item $|\psi_{0}\rangle$ is a unit vector in the Hilbert space ${\cal H}(Q)$ (the {\em
initial quantum state});

\item $\delta: S\times \Sigma\rightarrow S$ is
a map (the {\em classical transition map});

\item $\mathbb{U}=\{U_{s\sigma}\}_{s\in S,\sigma\in \Sigma}$ where $U_{s\sigma}:{\cal
H}(Q)\rightarrow {\cal H}(Q)$ is a unitary operator for each $s$
and $\sigma$ (the {\em quantum transition operator} at $s$ and $\sigma$);

\item ${\cal M}=\{{\cal M}_s\}_{s\in S}$  where each ${\cal M}_s$ is a projective measurement over ${\cal H}(Q)$ with outcomes in $\Gamma$ (the {\em measurement operator at} $s$).
\end{itemize}
\end{Df}

Hence, each ${\cal M}_s= \{P_{s,\gamma}\}_{\gamma\in \Gamma}$ such that
$\sum_{\gamma\in \Gamma}P_{s,\gamma}=I$ and
$P_{s,\gamma}P_{s,\gamma'}=\left\{\begin{array}{ll}P_{s,\gamma},&
\gamma=\gamma',\\
O,& \gamma\not=\gamma'.
\end{array}
\right.$
Furthermore, if the machine is
in classical state $s$ and quantum state $|\psi\rangle$ after
reading the input string, then  $\|P_{s,\gamma}|\psi\rangle\|^{2}$
is the probability of the machine producing
outcome $\gamma$ on that input.

\begin{Rm}\em
Map
$\delta$ can be extended to a map $\delta^{*}:
\Sigma^{*}\rightarrow S$ as usual. That is,
$\delta^{*}(s,\epsilon)=s$; for any string $x\in\Sigma^{*}$ and
any $\sigma\in \Sigma$, $\delta^{*}(s,\sigma x)=
\delta^{*}(\delta(s,\sigma),x)$.

\end{Rm}

\begin{Rm}\em  A specially interesting case of the above definition is when $\Gamma=\{a, r\}$,
    where $a$ denotes  {\em accepting} and $r$ denotes {\em rejecting}.
Then, ${\cal M}=\{\{P_{s,a},P_{s,r}\}:s\in S\}$ and, for each
$s\in S$, $P_{s,a}$ and $P_{s,r}$ are two projectors such that $P_{s,a}+P_{s,r}=I$ and $P_{s,a}P_{s,r}=
O$. In this case, ${\cal A}$ is an
acceptor of languages over $\Sigma$.
\end{Rm}

For the sake of convenience, we denote the map $\mu : \Sigma^{*}
\rightarrow S$, induced by $\delta$, as
$\mu(x)=\delta^{*}(s_{0},x)$ for any string $x\in\Sigma^{*}$.

We further describe the computing process of ${\cal
A}=(S,Q,\Sigma, s_{0},|\psi_{0}\rangle, \delta,
\mathbb{U},{\cal M})$ for input string
$x=\sigma_{1}\sigma_{2}\cdots\sigma_{m}$ where $\sigma_{i}\in
\Sigma$ for $i=1,2,\cdots,m$.

The machine ${\cal A}$ starts at the
initial classical state $s_{0}$ and initial quantum state
$|\psi_{0}\rangle$. On reading the first symbol $\sigma_{1}$ of the input string, the states of the machine change as follows: the classical state becomes
$\mu(\sigma_{1})$; the quantum state becomes $U_{s_{0}\sigma_{1}}|\psi_{0}\rangle$.
Afterward, on reading $\sigma_{2}$, the machine changes its classical state to $\mu(\sigma_{1}\sigma_{2})$ and its quantum state to the result of applying $U_{\mu(\sigma_{1})\sigma_{2}}$ to
$U_{s_{0}\sigma_{1}}|\psi_{0}\rangle$.

The process continues
similarly by reading $\sigma_{3},\sigma_{4},\cdots,\sigma_{m}$ in succession.
Therefore, after reading $\sigma_{m}$, the classical state becomes
$\mu(x)$ and the quantum state is as follows:
\begin{equation}
U_{\mu( \sigma_{1}\cdots \sigma_{m-2}\sigma_{m-1} )\sigma_{m}}U_{\mu(\sigma_{1}\cdots \sigma_{m-3}\sigma_{m-2})\sigma_{m-1}}\cdots U_{\mu(\sigma_{1})\sigma_{2}}
U_{s_{0}\sigma_{1}}|\psi_{0}\rangle.
\end{equation}

Let ${\cal U}(Q)$ be the set of unitary operators on Hilbert space
${\cal H}(Q)$. For the sake of convenience, we denote the map
$v:\Sigma^{*}\rightarrow {\cal U}(Q)$ as: $v(\epsilon)=I$ and
\begin{equation}
v(x)=
U_{\mu( \sigma_{1}\cdots \sigma_{m-2}\sigma_{m-1} )\sigma_{m}}U_{\mu(\sigma_{1}\cdots \sigma_{m-3}\sigma_{m-2})\sigma_{m-1}}\cdots U_{\mu(\sigma_{1})\sigma_{2}}
U_{s_{0}\sigma_{1}} \label{v}
\end{equation}
for $x=\sigma_{1}\sigma_{2}\cdots\sigma_{m}$ where $\sigma_{i}\in
\Sigma$ for $i=1,2,\cdots,m$, and $I$ denotes the identity
operator on ${\cal H}(Q)$, indicated as before.

By means of the denotations $\mu$ and $v$, for any input string
$x\in\Sigma^{*}$, after ${\cal A}$ reading $x$, the classical
state is $\mu(x)$ and the quantum states $v(x)|\psi_{0}\rangle$.

Finally,  the probability $\Prob_{{\cal A},\gamma}(x)$ of machine
${\cal A}$ producing result $\gamma$ on input $x$  is as follows:
\begin{equation}
\Prob_{{\cal A},\gamma}(x)= \|P_{\mu(x),\gamma}v(x)|\psi_{0}\rangle\|^{2}.
\end{equation}

In particular, when ${\cal A}$ is thought of as an acceptor of
languages over $\Sigma$ ($\Gamma=\{a,r\}$), we obtain the
probability $\Prob_{{\cal A},a}(x)$ for accepting $x$:
\begin{equation}
\Prob_{{\cal A},a}(x)= \|P_{\mu(x),a}v(x)|\psi_{0}\rangle\|^{2}.
\end{equation}

For intuition, we depict the above process in Figure~\ref{fig:1qfacdyn}.

\begin{figure}[htbp]
    \centering
        \ifx\JPicScale\undefined\def\JPicScale{0.9}\fi
        \unitlength \JPicScale mm

        \begin{picture}(155,70)(10,20)
        \linethickness{0.3mm}
        \put(40,50){\line(0,1){10}}
        \linethickness{0.3mm}
        \put(50,50){\line(0,1){10}}
        \linethickness{0.3mm}
        \put(10,40){\line(1,0){10}}
        \put(20,30){\line(0,1){10}}
        \put(10,30){\line(1,0){10}}
        \put(10,30){\line(0,1){10}}
        \linethickness{0.3mm}
        \put(10,80){\line(1,0){10}}
        \put(20,70){\line(0,1){10}}
        \put(10,70){\line(1,0){10}}
        \put(10,70){\line(0,1){10}}
        \put(15,35){\makebox(0,0)[cc]{$s_0$}}

        \put(15,75){\makebox(0,0)[cc]{\small$|\psi_0\rangle$}}

        \linethickness{0.3mm}
        \put(30,80){\line(1,0){10}}
        \put(40,70){\line(0,1){10}}
        \put(30,70){\line(1,0){10}}
        \put(30,70){\line(0,1){10}}
        \linethickness{0.3mm}
        \put(20,35){\line(1,0){10}}
        \put(30,35){\vector(1,0){0.12}}
        \linethickness{0.3mm}
        \put(35,45){\line(0,1){5}}
        \linethickness{0.3mm}
        \put(25,45){\line(1,0){10}}
        \linethickness{0.3mm}
        \put(25,35){\line(0,1){10}}
        \put(25,35){\vector(0,-1){0.12}}
        \linethickness{0.3mm}
        \put(40,35){\line(1,0){10}}
        \put(50,35){\vector(1,0){0.12}}
        \linethickness{0.3mm}
        \put(30,40){\line(1,0){10}}
        \put(40,30){\line(0,1){10}}
        \put(30,30){\line(1,0){10}}
        \put(30,30){\line(0,1){10}}
        \put(35,35){\makebox(0,0)[cc]{$s_1$}}

        \put(25,30){\makebox(0,0)[cc]{$\delta$}}

        \linethickness{0.3mm}
        \put(20,75){\line(1,0){10}}
        \put(30,75){\vector(1,0){0.12}}
        \linethickness{0.3mm}
        \linethickness{0.3mm}
        \put(15,40){\line(0,1){25}}
        \linethickness{0.3mm}
        \put(15,65){\line(1,0){10}}
        \linethickness{0.3mm}
        \put(25,65){\line(0,1){10}}
        \put(25,75){\vector(0,1){0.12}}
        \linethickness{0.3mm}
        \put(35,60){\line(0,1){5}}
        \linethickness{0.3mm}
        \put(25,65){\line(1,0){10}}
        \put(25,80){\makebox(0,0)[cc]{$U$}}

        \linethickness{0.3mm}
        \put(100,40){\line(1,0){10}}
        \put(110,30){\line(0,1){10}}
        \put(100,30){\line(1,0){10}}
        \put(100,30){\line(0,1){10}}
        \linethickness{0.3mm}
        \put(110,35){\line(1,0){10}}
        \put(120,35){\vector(1,0){0.12}}
        \linethickness{0.3mm}
        \put(125,45){\line(0,1){5}}
        \linethickness{0.3mm}
        \put(115,45){\line(1,0){10}}
        \linethickness{0.3mm}
        \put(115,35){\line(0,1){10}}
        \put(115,35){\vector(0,-1){0.12}}
        \linethickness{0.3mm}
        \put(120,40){\line(1,0){10}}
        \put(130,30){\line(0,1){10}}
        \put(120,30){\line(1,0){10}}
        \put(120,30){\line(0,1){10}}
        \linethickness{0.3mm}
        \linethickness{0.3mm}
        \put(120,50){\line(0,1){10}}
        \linethickness{0.3mm}
        \put(100,80){\line(1,0){10}}
        \put(110,70){\line(0,1){10}}
        \put(100,70){\line(1,0){10}}
        \put(100,70){\line(0,1){10}}
        \linethickness{0.3mm}
        \put(120,80){\line(1,0){10}}
        \put(130,70){\line(0,1){10}}
        \put(120,70){\line(1,0){10}}
        \put(120,70){\line(0,1){10}}
        \linethickness{0.3mm}
        \put(110,75){\line(1,0){10}}
        \put(120,75){\vector(1,0){0.12}}
        \linethickness{0.3mm}
        \put(105,65){\line(1,0){10}}
        \linethickness{0.3mm}
        \put(115,65){\line(0,1){10}}
        \put(115,75){\vector(0,1){0.12}}
        \linethickness{0.3mm}
        \put(125,60){\line(0,1){5}}
        \linethickness{0.3mm}
        \put(115,65){\line(1,0){10}}
        \linethickness{0.3mm}
        \put(105,40){\line(0,1){25}}
        \put(35,75){\makebox(0,0)[cc]{\small$|\psi_1\rangle$}}

        \put(105,75){\makebox(0,0)[cc]{\small$|\psi_{\textrm{\tiny$n\!\!-\!\!1\!$}}\rangle$}}

        \put(125,75){\makebox(0,0)[cc]{\small$|\psi_n\rangle$}}

        \put(105,35){\makebox(0,0)[cc]{$s_{n-1}$}}

        \put(125,35){\makebox(0,0)[cc]{$s_n$}}

        \put(115,30){\makebox(0,0)[cc]{$\delta$}}

        \put(115,80){\makebox(0,0)[cc]{$U$}}

        \put(135,80){\makebox(0,0)[cc]{$M$}}

        \linethickness{0.3mm}
        \put(40,75){\line(1,0){10}}
        \put(50,75){\vector(1,0){0.12}}
        \linethickness{0.3mm}
        \put(90,35){\line(1,0){10}}
        \put(100,35){\vector(1,0){0.12}}
        \linethickness{0.3mm}
        \put(90,75){\line(1,0){10}}
        \put(100,75){\vector(1,0){0.12}}
        \linethickness{0.3mm}
        \put(30,60){\line(1,0){70}}
        \linethickness{0.3mm}
        \put(30,50){\line(0,1){10}}
        \linethickness{0.3mm}
        \put(30,50){\line(1,0){70}}
        \linethickness{0.3mm}
        \put(110,60){\line(1,0){20}}
        \linethickness{0.3mm}
        \put(130,50){\line(0,1){10}}
        \linethickness{0.3mm}
        \put(110,50){\line(1,0){20}}
        \put(75,35){\makebox(0,0)[cc]{...}}

        \put(75,55){\makebox(0,0)[cc]{...}}

        \put(75,75){\makebox(0,0)[cc]{...}}

        \put(115,55){\makebox(0,0)[cc]{...}}

        \put(35,55){\makebox(0,0)[cc]{$\sigma_0$}}

        \put(45,55){\makebox(0,0)[cc]{$\sigma_1$}}

        \put(125,55){\makebox(0,0)[cc]{$\sigma_n$}}

        \linethickness{0.3mm}
        \put(130,35){\line(1,0){5}}
        \linethickness{0.3mm}
        \put(135,35){\vector(0,1){40}}
        \linethickness{0.3mm}
        \put(130,75){\line(1,0){10}}
        \put(140,75){\vector(1,0){0.12}}
        \linethickness{0.3mm}
        \put(140,75){\line(0,1){5}}
        \linethickness{0.3mm}
        \put(140,80){\line(1,0){5}}
        \put(157,80){\makebox(0,0)[cc]{Accept}}

        \put(157,70){\makebox(0,0)[cc]{Reject}}

        \linethickness{0.3mm}
        \put(145,80){\line(1,0){5}}
        \put(150,80){\vector(1,0){0.12}}
        \linethickness{0.3mm}
        \put(140,70){\line(0,1){5}}
        \linethickness{0.3mm}
        \put(140,70){\line(1,0){10}}
        \put(150,70){\vector(1,0){0.12}}
        \end{picture}

    \caption{1QFAC dynamics as an acceptor of languages}
    \label{fig:1qfacdyn}
\end{figure}
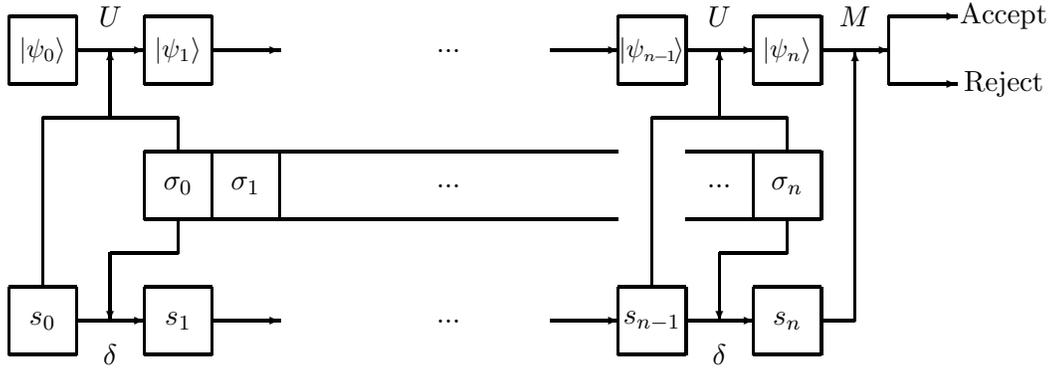

\begin{Rm}\em  If a 1QFAC ${\cal A}$ has only one classical state,
then ${\cal A}$ reduces to an MO-1QFA \cite{MC00}. Therefore, the
set of languages accepted by  1QFAC with only one classical state
is a proper subset of regular languages (exactly, the languages
whose syntactic monoid is a group  \cite{BP02}). However, we prove that 1QFAC
can accept all regular languages with no error.
\end{Rm}

\begin{Pp} Let $\Sigma$ be a finite set. Then each regular language over  $\Sigma$ that is accepted by a minimal DFA of $k$ states
is also accepted by some 1QFAC with no error and with 1 quantum basis state and $k$ classical states.
\end{Pp}

\begin{proof}
Let $L\subseteq \Sigma^{*}$ be a regular language. Then there
exists a DFA
$M=(S,\Sigma,\delta,s_{0},F)$ accepting $L$, where, as usual, $S$
is a finite set of states, $s_{0}\in S$ is an initial state,
$F\subseteq Q$ is a set of accepting states, and $\delta: Q\times
\Sigma\rightarrow Q$ is the transition function. We construct a
1QFAC ${\cal A}=(S,Q,\Sigma,\Gamma, s_{0},|\psi_{0}\rangle,
\delta, \mathbb{U}, {\cal M})$ accepting $L$ without error, where
$S$,  $\Sigma$, $s_{0}$, and $\delta$ are the same as those in
$M$, and, in addition, $\Gamma=\{a,r\}$, $Q=\{0\}$,
$|\psi_{0}\rangle=|0\rangle$, $\mathbb{U}=\{U_{s\sigma}: s\in
S,\sigma\in\Sigma\}$ with $U_{s\sigma}=I$ for all $s\in S$ and
$\sigma\in\Sigma$, ${\cal M}=\{\{P_{s,a},P_{s,r}\}:s\in S\}$
assigned as: if $s\in F$, then $P_{s,a}=|0\rangle\langle 0|$ and
$P_{s,r}=O$ where $O$ denotes the zero operator as before; otherwise,
$P_{s,a}=O$ and $P_{s,r}=|0\rangle\langle 0|$.

By the above definition of 1QFAC ${\cal A}$, it is easy to check
that the language accepted by ${\cal A}$ with no error is exactly
$L$.
\end{proof}

\begin{Rm}\em  \label{1QFACL} For any regular language $L$ over $\{0,1\}$ accepted by a $k$ state DFA, it was proved that there exists a 1QFACL accepting $L$ with no error and with $3k$ classical states ($3k$ is the number of states of its minimal DFA accepting the control language)  and 3 quantum basis states \cite{MP06}. Here, for 1QFAC, we require only $k$ classical states and 1 quantum basis states. Therefore, in this case, 1QFAC have better state complexity than 1QFACL.
\end{Rm}

\begin{Rm}\em \label{rm:alterantive}
On the other hand, any language accepted by a 1QFAC
is regular. We can prove this result in detail, based on a well-know idea for one-way probabilistic automata by Rabin \cite{Rab63}, that was already applied for MM-1QFA by Kondacs and Watrous \cite{KW97} as well as for MO-1QFA by Brodsky and Pippenger \cite{BP02}. However, the process is much longer and further results are needed, since both classical and quantum states are involved in 1QFAC. Another possible approach is based on topological automata \cite{Boz03,Jea07}.  However, in next section we obtain this result while studying the state complexity of 1QFAC and so we postpone the proof of regularity to the next section.
\end{Rm}

\section{State complexity of 1QFAC}
State complexity of classical finite automata has been a hot research subject with important practical applications \cite{RS97,Yu98}. In this section,  we consider this problem for 1QFAC. First, we  prove a lower bound on the state complexity of 1QFAC which states that 1QFAC are at most exponentially more concise than DFA. Second, we show that our bound is tight by giving some languages that witness the exponential advantage of 1QFAC over DFA. Particularly, these languages can not be accepted by any MO-1QFA, MM-1QFA or multi-letter 1QFA.
\subsection{On the lower bound for 1QFAC}
In this section,  we prove a lower bound for the state complexity of 1QFAC which states that 1QFAC are at most exponentially more concise than DFA. Also, we show that the languages accepted by 1QFAC with bounded error are regular. Some examples given in the next subsection shows that our lower bound is tight.

Given a 1QFAC ${\cal A}=(S,Q,\Sigma,\Gamma, s_{0},|\psi_{0}\rangle, \delta,
\mathbb{U}, {\cal M})$, we reformulate it as a mathematical model $({\cal H}, |\phi_0\rangle, \{M(\sigma): \sigma\in\Sigma\}, \{P_\gamma: \gamma\in\Gamma\})$ that is useful to our discussion,
where
\begin{itemize}
  \item  ${\cal H}={\cal H}(S)\otimes {\cal H}(Q)$;
   \item $|\phi_0\rangle=|s_0\rangle|\psi_0\rangle$;
  \item $M(\sigma)=\sum_{s\in S}|\delta(s,\sigma)\rangle\langle s|\otimes U_{s\sigma}$ for $\sigma\in\Sigma$;
  \item $P_\gamma=\sum_{s\in S}|s\rangle\langle s|\otimes P_{s\gamma}$ for each $\gamma\in\Gamma$.
\end{itemize}
It is easy to verify that \begin{equation}\Prob_{{\cal A},\gamma}(x)=\|P_\gamma M(x)|\phi_0\rangle\|^2 \label{Prob}\end{equation} for each $\gamma\in \Gamma$ and $x\in\Sigma^*$, where $M(x_1\cdots x_n)=M(x_n)\cdots M(x_1)$.
Furthermore, we let
\begin{equation} {\cal V}=\{|\phi_x\rangle: |\phi_x\rangle=M(x)|\phi_0\rangle, x\in\Sigma^*\}. \label{df-v}\end{equation}
Then we have the following result.
\begin{Lm}\label{Lm-v}It holds that
\begin{itemize}
  \item [(i)] each $|\phi\rangle\in{\cal V}$ has the form $|\phi\rangle=|s\rangle|\psi\rangle$ where $s\in S$ and $|\psi\rangle\in {\cal H}(Q)$;
  \item [(ii)] $\||\phi\rangle\|^2=1$ for all $|\phi\rangle\in{\cal V}$;
  \item [(iii)] $\|M(x)|\phi_1\rangle-M(x)|\phi_2\rangle\|\leq\sqrt{2}\||\phi_1\rangle-|\phi_2\rangle\|$ for all $x\in\Sigma^*$.
\end{itemize}
\end{Lm}
\begin{proof} Items (i) and (ii) are easy to be verified and here we omit the proof of them. In the following, we prove item (iii) in detail. Let $|\phi_i\rangle=|s_i\rangle|\psi_i\rangle$ and $|\phi'_i\rangle=M(x)|\phi_i\rangle=|s'_i\rangle|\psi'_i\rangle$ for $i=1,2$ and $x\in \Sigma^*$, where $s_i, s'_i\in S$ and $|\psi_i\rangle, |\psi'_i\rangle\in {\cal H}(Q)$. The  discussion is divided into two cases.

Case (a): $|s_1\rangle=|s_2\rangle$. In this case it necessarily holds that $|s'_1\rangle=|s'_2\rangle$ and furthermore we have
\begin{equation}
\||\phi'_1\rangle-\phi'_2\rangle\|=\||\psi'_1\rangle-|\psi'_2\rangle\|=\||\psi_1\rangle-|\psi_2\rangle\|=\||\phi_1\rangle-|\phi_2\rangle\|,
\end{equation}
where the first and third equations hold because of $\||\alpha\rangle|\beta\rangle\|=\||\alpha\rangle\|.\||\beta\rangle\|$ and the second holds since $|\psi'_1\rangle$ and $|\psi'_2\rangle$ are obtained by performing the same unitary operation on $|\psi_1\rangle$ and $|\psi_2\rangle$, respectively.

Case (b): $|s_1\rangle\neq|s_2\rangle$. First it holds that $\||\phi_1\rangle-|\phi_2\rangle\|=\sqrt{2}$. Indeed, let $|\psi_1\rangle=\sum_i\alpha_i|i\rangle$ and $|\psi_2\rangle=\sum_i\beta_i|i\rangle$. Then we have
\begin{eqnarray}\||\phi_1\rangle-|\phi_2\rangle\|&=&\||s_1\rangle|\psi_1\rangle-|s_2\rangle|\psi_2\rangle\|\\
&=&\left\|\sum_i\alpha_i|s_1\rangle|i\rangle+ \sum_i(-\beta_i)|s_2\rangle|i\rangle\right\|\\
&=& \left(\sum_i |\alpha_i|^2+\sum_i |\beta_i|^2\right)^{\frac{1}{2}}\\
&=&\sqrt{\||\psi_1\rangle\|^2+\||\psi_1\rangle\|^2}\\
&=&\sqrt{2}
\end{eqnarray}
Therefore, \begin{eqnarray}
\||\phi'_1\rangle-\phi'_2\rangle||&=&\||s'_1\rangle|\psi'_1\rangle-|s'_2\rangle|\psi'_2\rangle\|\\
&=&\left\{
     \begin{array}{ll}
       \|\psi'_1\rangle-|\psi'_2\rangle\|, & \hbox{if $s'_1= s'_2$;} \\
       \sqrt{2}, & \hbox{else.}
     \end{array}
   \right.
\end{eqnarray}
Note that $ \|\psi'_1\rangle-|\psi'_2\rangle\|\leq 2=\sqrt{2}\||\phi_1\rangle-|\phi_2\rangle\|$.

 In summary, item (iii) holds in any case.
\end{proof}

Next we present another lemma which is critical for obtaining the lower bound on 1QFAC.
\begin{Lm} \label{Lm-s}
Let ${\cal V}_\theta\subseteq \mathbb{C}^n$ such that $\||\phi_1\rangle-|\phi_2\rangle\|\geq \theta$ for any two elements $|\phi_1\rangle,|\phi_2\rangle\in {\cal V}_\theta$. Then ${\cal V}_\theta$ is a finite set containing  $k(\theta)$ elements where $k(\theta)\leq(1+\frac{2}{\theta})^{2n}$.
\end{Lm}
\begin{proof} Arbitrarily choose an element $|\phi\rangle\in {\cal V}_\theta$. Let $U(|\phi\rangle, \frac{\theta}{2})=\{|\chi\rangle: \||\chi\rangle-|\phi\rangle\|\leq\frac{\theta}{2}\}$, i.e., a sphere centered at $|\phi\rangle$ with the radius $\frac{\theta}{2}$. Then  all these spheres do not intersect pairwise except for  their surface, and all of them are contained in a large sphere centered at $(0,0,\cdots,0)$ with the radius $1+\frac{\theta}{2}$. The volume of a sphere of a radius $r$ in $\mathbb{C}^n$ is $cr^{2n}$ where $c$ is a constant. Note that $\mathbb{C}^n$ is an $n$-dimensional complex space and each element from it can be represented by an element of $\mathbb{R}^{2n}$. Therefore, it holds that
\begin{equation}
k(\theta)\leq\frac{c(1+\frac{\theta}{2})^{2n}}{c(\frac{\theta}{2})^{2n}}=(1+\frac{2}{\theta})^{2n}.
\end{equation}
\end{proof}
\noindent
Below we recall a result that will be used later on (c.f. Lemma 8 in \cite{Ya03} for a complete proof).
\begin{Lm}\label{Lm-p}
For any two elements $|\phi\rangle, |\varphi\rangle\in \mathbb{C}^n$ with $\||\phi\rangle\|\leq c$ and $\||\varphi\rangle\|\leq c$, it holds that $\left | \|P|\phi\rangle\|^2- \|P|\varphi\rangle\|^2   \right|\leq c\||\phi\rangle-|\varphi\rangle\|$ where $P$ is a projective operator on $\mathbb{C}^n$.
\end{Lm}

Given a language $L\subseteq\Sigma^*$, define an equivalence relation ``$\equiv_L$'' as: for any  $x,y\in\Sigma^*$, $x\equiv_L y$ if
for any $z\in\Sigma^*$, $xz\in L$ iff $yz\in
L$. If $x,y$ do not satisfy the equivalence relation, we denote it by $x\not\equiv_L y$. Then  the set $\Sigma^*$ is partitioned into some equivalence classes by the equivalence relation ``$\equiv_L$''. In the following we recall a well-known result that will be used in the sequel.
\begin{Lm}[Myhill-Nerode theorem \cite{HU79}]\label{MN-Th}
A language $L\subseteq\Sigma^*$ is regular if and only if the number of equivalence classes induced by the equivalence relation ``$\equiv_L$'' is finite. Furthermore, the number of  equivalence classes equals to the state number of the minimal DFA accepting $L$.
\end{Lm}

Now we are ready to present our main result.
\begin{thm} \label{bound}
 If $L$ is accepted by a 1QFAC ${\cal M}$ with bounded error, then $L$ is regular and it holds that $kn=\Omega( log ~m)$ where $k$ and $n$ denote  numbers of classical states and quantum basis states of ${\cal M}$, respectively, and $m$ is the state number of the minimal DFA accepting $L$.
\end{thm}
\begin{proof}
 Let ${\cal V}'\subseteq {\cal V}$ (where ${\cal V}$ is given in Eq. (\ref{df-v})) satisfying  for  any two elements $|\phi_{x}\rangle, |\phi_{y}\rangle\in {\cal V}'$  it holds that  $|\phi_{x}\rangle\not=|\phi_{y}\rangle\Leftrightarrow x\not\equiv_L y$.
Then for two different elements $|\phi_{x}\rangle, |\phi_{y}\rangle\in {\cal V}'$ there exists $z\in\Sigma^*$ satisfying $xz\in L$ whereas $yz\not\in L$ (or $xz\not\in L$ whereas $yz\in L$). That is
\begin{eqnarray} \Prob_{{\cal A},a}(xz)&=&||P_a M(z)|\phi_x\rangle||^2\geq\lambda+\epsilon,\\
\Prob_{{\cal A},a}(yz)&=&||P_a M(z)|\phi_y\rangle||^2\leq\lambda-\epsilon\end{eqnarray}
for some $\lambda\in(0,1]$ and $\epsilon>0$.
Therefore we have
\begin{eqnarray} \sqrt{2}\||\phi_x\rangle-|\phi_y\rangle\|&\geq &\|M(z)|\phi_x\rangle-M(z)|\phi_y\rangle\|\\
&\geq &|\Prob_{{\cal A},a}(xz)-\Prob_{{\cal A},a}(yz)|\\
&\geq &2\epsilon\end{eqnarray}
where the first inequality follows from Lemma \ref{Lm-v} and the second follows from Lemma \ref{Lm-p}.
  In summary, we obtain that two different elements $|\phi_{x}\rangle$ and $|\phi_{y}\rangle$ from ${\cal V}'$ satisfy $\|\phi_x\rangle-\phi_y\rangle\|\geq \sqrt{2}\epsilon.$  Therefore, according to Lemma \ref{Lm-s}, we have that the number $|{\cal V}'|$ of elements in ${\cal V}'$ satisfies $|{\cal V}'|\leq(1+\frac{\sqrt{2}}{\epsilon})^{2kn} $, which means that the number of equivalence classes induced by  the equivalence relation ``$\equiv_L$'' is upper bounded by $(1+\frac{\sqrt{2}}{\epsilon})^{2kn}$. Therefore, by Lemma \ref{MN-Th} we have completed the proof. \end{proof}

When the number of classical states equals one in a 1QFAC  ${\cal M}$,   ${\cal M}$ exactly reduces to an MO-1QFA. Therefore,
as a corollary, we can obtain a precise relationship between the numbers of states for MO-1QFA and DFA that was also derived by Ablayev and Gainutdinova \cite{AG00} (though there are two cases in \cite{AG00} by dividing $\epsilon$ into two intervals, from our proof we find it is not necessary).

\begin{Co}
If $L$ is accepted by an MO-1QFA ${\cal M}$ with bounded error, then $L$ is regular and it holds that $n=\Omega( \log m)$ where  $n$ denotes the number of quantum basis states of ${\cal M}$, and $m$ is the state number of the minimal DFA accepting $L$.

\end{Co}

\subsection{The lower bound is tight}

Although 1QFAC accept only regular languages as  DFA, 1QFAC can accept some languages with essentially less number of states than DFA and these languages  cannot  be accepted by any MO-1QFA or MM-1QFA or multi-letter 1QFA. In this section, our purpose is to prove these claims, and we also obtain that the lower bound in Theorem \ref{bound} is tight.

First, we establish a technical result concerning the acceptability by 1QFAC of languages resulting from set operations on languages accepted by MO-1QFA and by DFA.

\begin{Lm} \label{operation}

Let $\Sigma$ be a finite alphabet. Suppose that the language $L_{1}$ over  $\Sigma$ is accepted by a minimal DFA with $n_{1}$ states and  the language $L_{2}$ over  $\Sigma$ is accepted by an MO-1QFA with $n_{2}$ quantum basis states with bounded error $\epsilon$. Then  the intersection $L_{1}\cap L_{2}$, union $L_{1}\cup L_{2}$, differences $L_{1}\setminus L_{2}$ and $L_{2}\setminus L_{1}$ can be accepted by some 1QFAC  with  $n_{1}$ classical states and $n_{2}$ quantum basis states with bounded error $\epsilon$.
\end{Lm}

\begin{proof} Let $A_1=(S,\Sigma,\delta,s_0,F)$ be a minimal DFA accepting $L_1$, and let $A_2=(Q, \Sigma, |\psi_{0}\rangle, \\ \{U(\sigma)\}_{\sigma\in\Sigma},Q_{acc})$ be an MO-1QFA accepting $L_2$, where $s_0\in S$ is the initial state, $\delta$ is the transition function, and $F\subseteq S$ is a finite subset denoting accepting states; the symbols in $A_2$ are the same as those in the definition of MO-1QFA as above.

Then by $A_1$ and $A_2$ we define a 1QFAC ${\cal A}=(S,Q,\Sigma,\Gamma, s_{0},|\psi_{0}\rangle, \delta,
\mathbb{U}, {\cal M})$  accepting $L_{1}\cap L_{2}$, where  $S,Q,\Sigma, s_{0},|\psi_{0}\rangle, \delta$ are the same as those in $A_1$ and $A_2$, $\Gamma=\{a,r\}$,  $\mathbb{U}=\{U_{s\sigma}=U(\sigma): s\in S, \sigma\in \Sigma\}$, and ${\cal M}=\{M_{s}:s\in S\}$ where $M_s=\{P_{s,a},P_{s,r} \}$ and
\[
P_{s,a}=\left\{\begin{array}{ll}
\sum_{p\in Q_{acc}}|p\rangle\langle p|, & s\in F;\\
O,& s\not\in F,
\end{array}
\right.
\] where $O$ denotes the zero operator, and
$P_{s,r}=I-P_{s,a}$ with $I$ being the identity operator.

According to the above definition of 1QFAC, we easily know that, for any string $x\in \Sigma^*$, if $x\in L_1$ then the accepting probability of
1QFAC ${\cal A}$ is equal to the accepting probability of MO-1QFA $A_2$; if $x\not\in L_1$ then the accepting probability of
1QFAC ${\cal A}$ is zero. So, 1QFAC ${\cal A}$ accepts the intersection $L_{1}\cap L_{2}$.

Similarly, we can construct the other three 1QFAC accepting the union $L_{1}\cup L_{2}$, differences $L_{1}\setminus L_{2}$, and $L_{2}\setminus L_{1}$, respectively. Indeed, we only need define different measurements in these 1QFAC. If we construct 1QFAC accepting $L_{1}\cup L_{2}$, then
\[P_{s,a}=\left\{\begin{array}{ll}
I,&  s\in F;\\
\sum_{p\in Q_{acc}}|p\rangle\langle p|, & s\not\in F.
\end{array}
\right.
\]
If we construct 1QFAC accepting $L_{1}\setminus L_{2}$, then
\[
P_{s,a}=\left\{\begin{array}{ll}
\sum_{p\in Q\setminus Q_{acc}}|p\rangle\langle p|, & s\in F;\\
O,& s\not\in F.
\end{array}
\right.
\]
If we construct 1QFAC accepting $L_{2}\setminus L_{1}$, then
\[
P_{s,a}=\left\{\begin{array}{ll}
\sum_{p\in Q\setminus Q_{acc}}|p\rangle\langle p|, & s\not\in F;\\
O,& s\in F.
\end{array}
\right.
\]

\end{proof}

Now we consider a regular language $$L^0(m)=\{w0: w\in \{0,1\}^*, |w0|=km,k=1,2,3,\cdots \}.$$  Clearly, the minimal classical DFA accepting $L^0(m)$ has  $m+1$ states, as depicted in Figure~\ref{fig:dfal(m)}.

\begin{figure}[htbp]%
$$\xymatrix  {
 *+++[o][F-]{q_0}\ar@<-1ex>[r]^{0,1}
&*+++[o][F-]{q_1}\ar@<-1ex>[r]^{0,1}
&*+++[o][F-]{q_2}\ar@<-1ex>[r]^{0,1}
&*\txt{ \ \ \ ... \ \ \ }\ar@<-1ex>[r]^{0,1}
&*++[o][F-]{q_{\textrm{\tiny $m\!\!-\!\!1$}}}\ar@/_2pc/@(ul,ur)[llll]_1\ar@<-1ex>[r]^{0}
&*+++[o][F=]{q_m}\ar@/^3pc/@(dr,dl)[llll]^{0,1}
}$$
\caption{DFA accepting $L^0(m)$}%
\label{fig:dfal(m)}%
\end{figure}
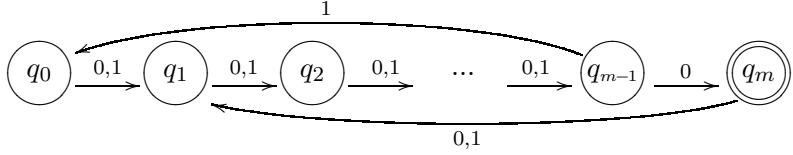

Indeed,  neither MO-1QFA nor MM-1QFA can accept
$L^0(m)$. We can easily verify this result by employing a lemma from \cite{BP02,GK02}. That is,

 \begin{Lm} [\cite{BP02,GK02}] \label{construction}  Let $L$ be a regular language, and let $M$ be its minimal DFA containing the construction in Figure 3, where states $p$ and $q$ are distinguishable (i.e., there exists a string  $z$ such that either $\delta(p,z)$ or  $\delta(q,z)$ is an accepting state). Then, $L$ can not be accepted by MM-1QFA.

 \end{Lm}

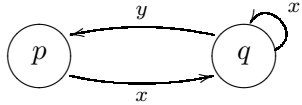
\begin{figure}[htbp]%
$$\xymatrix  {
 *+++[o][F-]{p}\ar@<-1ex>@/_/[rr]_{x}
&&*+++[o][F-]{q}\ar@<-1ex>@/_/[ll]_{y}\ar@(r,u)[]_{x}
}$$
\caption{Construction not accepted by an MM-1QFA}%
\label{fig:construction}%
\end{figure}

 \begin{Pp} Then neither MO-1QFA nor MM-1QFA can accept
 $L^0(m)$.

 \end{Pp}

 \begin{proof}
 Indeed, it suffices to show that no MM-1QFA can accept $L^0(m)$ since the  languages accepted by MO-1QFA are also accepted by MM-1QFA \cite{AF98,BP02,BMP03}.  By  Lemma \ref{construction},  we know that $L^0(m)$ can not be accepted by any MM-1QFA  since its minimal DFA (see Figure 2) contains such a construction: For example, we can take $p=q_0, q=q_m, x=0^m, y=0^{m-1}1, z=\epsilon$.
\end{proof}

Let us recall an important result from \cite{AF98}.

\begin{Pp}[\cite{AF98}] \label{Lp} Let the language $L_{p}=\{a^i: i \hskip 2mm  \textrm{is divisible by} \hskip 2mm p\}$ where $p$ is a prime number. Then for any $\varepsilon>0$, there exists an MM-1QFA with $O(\log(p))$ states such that for any $x\in L_{p}$, $x$ is accepted with no error, and the  probability for accepting $x\not\in  L_{p}$ is smaller than  $\varepsilon$.
 \end{Pp}

Indeed, from the proof of Proposition \ref{Lp} by \cite{AF98}, also as Ambainis and Freivalds pointed out in \cite{AF98} (before Section 2.2 in \cite{AF98}), Proposition \ref{Lp} holds for MO-1QFA as well.

Clearly, by the same technique as the proof of Proposition \ref{Lp} \cite{AF98}, then one can obtain that, by replacing $L_{p}$ with $L(m)=\{w: w\in \{0,1\}^*, |w|=km,k=1,2,3,\cdots \}$ with $m$ being a prime number, Proposition \ref{Lp} still holds (by viewing all input symbols in $\{0,1\}$ as $a$).
By combining Proposition \ref{Lp} with Lemma \ref{operation}, we have the following corollary.

\begin{Co} Suppose that $m$ is a prime number. Then for any $\varepsilon>0$, there exists a 1QFAC with 2 classical states and $O(\log(m))$ quantum basis states such that for any $x\in L^0(m)$, $x$ is accepted with no error, and the probability for accepting $x\not\in  L^0(m)$ is smaller than  $\varepsilon$.
 \end{Co}
\begin{proof} Note  that  we have $$L^0(m)=L^0\cap L(m)$$ where  $L^0=\{w0: w\in \{0,1\}^*\}$ is accepted by a DFA (depicted in Figure \ref {fig:dfaw0}) with only two states and $L(m)$ can be accepted by an MO-1QFA with $O(\log(m))$  quantum basis states as shown in Proposition \ref{Lp}. Therefore, the result follows from Lemma \ref{operation}.
\end{proof}
\begin{figure}[htbp]%
$$\xymatrix  {
 *+++[o][F-]{q_0}\ar@<-1ex>@/_/[rr]_{0}\ar@(l,u)[]^{1}
&&*+++[o][F=]{q_1}\ar@<-1ex>@/_/[ll]_{1}\ar@(r,u)[]_{0}
}$$
\caption{DFA accepting $\{0,1\}^*0$}%
\label{fig:dfaw0}%
\end{figure}
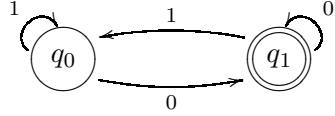

In summary, we have the following result.

 \begin{Pp}\label{prop:accbe} For any prime number $m\geq 2$, there exists a regular language $L^0(m)$ satisfying: (1) neither MO-1QFA nor MM-1QFA can accept $L^0(m)$; (2) the number of states in the minimal DFA accepting $L^0(m)$ is $m+1$; (3) for any  $\varepsilon>0$, there exists a 1QFAC with 2 classical states and $O(\log(m))$ quantum basis states such that for any $x\in L^0(m)$, $x$ is accepted with no error, and the probability for accepting $x\not\in  L^0(m)$ is smaller than  $\varepsilon$.

 \end{Pp}

\begin{Rm}\em  From the above proposition (see (2) and (3)) it follows that the lower bound given in Theorem \ref{bound} is tight, that is, attainable.
\end{Rm}

One should ask at this point whether similar results can be established for multi-letter 1QFA as proposed by Belovs et al. \cite{BRS07}.

Recall that $1$-letter 1QFA is exactly an MO-1QFA. Any  given
$k$-letter QFA can be simulated by some $k+1$-letter QFA. However, Qiu and Yu \cite{QY09}
proved that the contrary does not hold.  Belovs et al.
\cite{BRS07} have already showed that $(a+b)^{*}b$ can be accepted
by a 2-letter QFA but, as proved in \cite{KW97}, it cannot be
accepted by any MM-1QFA. On the other hand,  $a^*b^*$ can be accepted by MM-1QFA \cite{AF98} but it can not be accepted by any multi-letter 1QFA \cite{QY09}, and furthermore, there exists a regular language that can not be accepted by any MM-1QFA or  multi-letter 1QFA \cite{QY09}.

Let $\Sigma$ be an alphabet. For string $z=z_1\cdots z_n\in\Sigma^*$, consider the regular language  $$L_z=\Sigma^*z_1\Sigma^* z_2\Sigma^*\cdots \Sigma^* z_n \Sigma^*.$$ ($L_z$  belongs to piecewise testable set that was introduced by Simon \cite{Sim75} and studied in \cite{Per94}. Brodsky and Pippenger \cite{BP02} proved that $L_z$ can be accepted by an MM-1QFA with $2n+3$ states.) Let another regular language $L(m)=\{w: w\in \Sigma^*, |w|=km,k=1,2,\cdots\}$. Then the minimal DFA accepting $L_z$ needs $n+1$ states, and the minimal DFA accepting the intersection $L_z(m)$ of $L_z$ and $L(m)$  needs  $m
(n+1)$ states. We will prove that no multi-letter 1QFA can accept $L_z(m)$. Indeed,
the minimal DFA accepting $L_z(m)$ can be described by $A=(Q, \Sigma, \delta, q_0, F)$ where
$Q=\{S_{ij}: i=0, 1, \dots, n; j=1, 2, \dots, m\}$, $\Sigma=\{z_1, z_2, \dots, z_n\}$, $q_0=S_{01}$,
$F=\{S_{n1}\}$, and the transition function $\delta$ is defined as:
\begin{equation}\label{minimalDFA}
    \delta(S_{ij},\sigma)=\left\{
                            \begin{array}{ll}
                              S_{n,(j{\hskip -2mm} \mod m)+1}, & \textrm{if } i=n, \\
                              S_{i+1,(j{\hskip -2mm} \mod m)+1}, & \textrm{if } i\neq n \textrm{ and } \sigma=z_{i+1}, \\
                              S_{i,(j{\hskip -2mm} \mod m)+1}, &\textrm{if } i\neq n \textrm{ and } \sigma\neq z_{i+1}.
                            \end{array}
                          \right.
\end{equation}

The number of states of the minimal DFA accepting $L_z(m)$ is $m
(n+1)$.

For the sake of simplicity, we consider a special case: $m=2$, $n=1$, and $\Sigma=\{0,1\}$. Indeed, this case can also show the above problem as desired. So, we consider the following language:
\[
L_0(2)=\{w: w\in \{0,1\}^*0\{0,1\}^*, |w|=2k,k=1,2,\cdots\}.
\]

The minimal DFA accepting $L_0(2)$ above needs 4 states and its transition figure is depicted by Figure~\ref{fig:dfa0even} as follows.

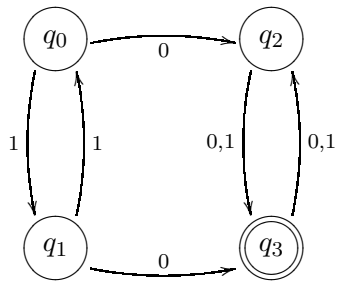
\begin{figure}[htbp]%
$$\xymatrix  {
 *+++[o][F-]{q_0}\ar@<-1ex>@/^/[rr]_{0}\ar@<-1ex>@/_/[dd]_{1}
&&*+++[o][F-]{q_2}\ar@<-1ex>@/_/[dd]_{0,1}\\\\
 *+++[o][F-]{q_1}\ar@<-1ex>@/_/[rr]^{0}\ar@<-1ex>@/_/[uu]_{1}
&&*+++[o][F=]{q_3}\ar@<-1ex>@/_/[uu]_{0,1}
}$$
\caption{DFA accepting $w\in \{0,1\}^*0\{0,1\}^*$ with $|w|$ even.}%
\label{fig:dfa0even}%
\end{figure}

We recall the definition of  F-construction and a proposition from \cite{BRS07}.

\begin{Df} [\cite{BRS07}]\em
A DFA  with state transition function $\delta$ is said to {\em contain
an F-construction} (see Figure~\ref{fig:fconstruction}) if there are non-empty words $t,z\in
\Sigma^{+}$ and two distinct states $q_{1},q_{2}\in Q$ such that
$\delta^{*}(q_{1},z)=\delta^{*}(q_{2},z)=q_{2}$,
$\delta^{*}(q_{1},t)=q_{1}$, $\delta^{*}(q_{2},t)=q_{2}$, where $\Sigma^+=\Sigma^*\backslash \{\epsilon\}$, $\epsilon$ denotes empty string.

\end{Df}

We can depict F-construction by Figure~\ref{fig:fconstruction}.

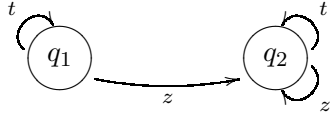
\begin{figure}[htbp]%
$$\xymatrix  {
 *+++[o][F-]{q_1}\ar@<-1ex>@/_/[rr]_{z}\ar@(l,u)[]^{t}
&&*+++[o][F-]{q_2}\ar@(r,u)[]_{t}\ar@(r,d)[]^{z}
}$$
\caption{F-Construction}%
\label{fig:fconstruction}%
\end{figure}

\begin{Lm} [\cite{BRS07}]
A language $L$ can be accepted by a multi-letter 1QFA with bounded
error if and only if the minimal DFA of $L$ does not contain  any
F-construction.
\end{Lm}

In Figure~\ref{fig:dfa0even}, there are  an F-construction: For example,  we  consider $q_0$ and $q_3$, and  strings $00$ and $11$, from the above proposition which shows that no multi-letter 1QFA  can accept $L_0(2)$.

Therefore, similarly to Proposition \ref{prop:accbe}, we have:

\begin{Pp}\em
 If  we have to restrict $m$ to be a prime number, then for any string $z$ with $|z|=n\geq 1$ there exists a regular language $L_z(m)$ that can not be accepted by any multi-letter 1QFA, but for every $\varepsilon$ there exists a 1QFAC ${\cal A}_m$  with $n+1$ classical states (independent of $m$) and $O(\log(m))$ quantum basis states such that if $x\in L_z(m)$, $x$ is accepted with no error, and the probability for accepting $x\not\in  L_z(m)$ is smaller than  $\varepsilon$. In contrast, the minimal DFA accepting $L_z(m)$ has $m(n+1)$ states.
\end{Pp}

\section{Determining the equivalence of 1QFAC}

In this section, we consider the equivalence problem of 1QFAC. For
any given 1QFAC ${\cal A}_{1}$ and 1QFAC ${\cal A}_{2}$ over the
same finite input alphabet $\Sigma$ and finite output alphabet
$\Gamma$, our purpose is to determine whether or not they are equivalent according to the following definition.

\begin{Df}\em
A 1QFAC ${\cal A}_1$ and another 1QFAC ${\cal A}_2$ over the same
input alphabet $\Sigma$ and output alphabet $\Gamma$ are said to
be {\em equivalent} (resp. {\em $t$-equivalent}) if $\Prob_{{\cal
A}_1,\gamma}(w)=\Prob_{{\cal A}_2,\gamma}(w)$ for any $w\in \Sigma^{*} $ (resp.
for any input string $w$ with $ |w|\leq t$) and any $\gamma\in\Gamma$.
\end{Df}

In the following, we will present a method to determine whether or not any two 1QFAC are equivalent.
 For readability, we recall a mathematical model which  is not an actual
computing model but generalizes many classical computing models,
 including probabilistic automata \cite{Rab63,Paz71} and deterministic finite automata \cite{HU79}.
\begin{Df}\label{def:BLM}\em
 A {\it bilinear  machine}
(BLM)  over the alphabet $\Sigma$ is a tuple $${\cal M}=( S,
\pi,\{M(\sigma)\}_{\sigma\in\Sigma},\eta),$$ where $S$ is a finite
state set with $|S|=n$, $\pi\in \mathbb { C}^{n\times 1}$, $\eta\in
\mathbb { C}^{1\times n}$ and $M(\sigma)\in \mathbb { C}^{n\times n}$ for
$\sigma\in \Sigma$.
\end{Df}
Associated to  a BLM ${\cal M}$, the {\it word function} $f_{\cal
M}:\Sigma^{*}\rightarrow \mathbb { C}$ is defined in the way: $f_{\cal
M}(w)=\eta M(w_m)\dots M(w_1)\pi$, where  $w=w_1\dots w_m\in
\Sigma^{*}$.

\begin{Df}\em
Two BLM ${\cal M}_1$ and
${\cal M}_2$ over the same alphabet $\Sigma$ are said to be
equivalent (resp. $k$-equivalent) if $f_{{\cal M}_1}(w)=f_{{\cal
M}_2}(w)$ for any $w\in \Sigma^{*} $ (resp. for any input string
$w$ with $ |w|\leq k$).
\end{Df}

As indicated in \cite{LQ08}, if we refer to \cite{Paz71,Tze92}, then we can  find that one can get a
general result as follows.
\begin{Pp}[\cite{Paz71,Tze92}] \label{BLM}
Two BLM  ${\cal A}_1$ and
${\cal A}_2$ with $n_1$ and $n_2$ states, respectively, are
equivalent if and only if they are $(n_1+n_2-1)$-equivalent.
Furthermore, there exists a polynomial-time algorithm running in
time $O((n_1+n_2)^4)$ that takes as input two BLM ${\cal A}_1$
and ${\cal A}_2$ and determines whether ${\cal A}_1$ and ${\cal
A}_2$ are equivalent.
\end{Pp}

Therefore, if we can simulate any 1QFAC by an equivalent BLM, then the equivalence problem of 1QFAC can be solved. Indeed, we can do  that by using the same technical treatments used in Section 3.1.  For the readability, below we recall these technical treatments.  Given a 1QFAC ${\cal A}=(S,Q,\Sigma,\Gamma, s_{0},|\psi_{0}\rangle, \delta,
\mathbb{U}, {\cal M})$, we construct
\begin{itemize}
  \item  ${\cal H}={\cal H}(S)\otimes {\cal H}(Q)$;
   \item $|\phi_0\rangle=|s_0\rangle|\psi_0\rangle$;
  \item $M(\sigma)=\sum_{s\in S}|\delta(s,\sigma)\rangle\langle s|\otimes U_{s\sigma}$ for $\sigma\in\Sigma$;
  \item $P_\gamma=\sum_{s\in S}|s\rangle\langle s|\otimes P_{s, \gamma}$ for each $\gamma\in\Gamma$.
\end{itemize}
 By using these notations, we have  \begin{equation}\Prob_{{\cal A},\gamma}(x)=\|P_\gamma M(x)|\phi_0\rangle\|^2 \label{Prob}\end{equation} for each $\gamma\in \Gamma$ and $x\in\Sigma^*$, where $M(x_1\cdots x_n)=M(x_n)\cdots M(x_1)$.  In the above construction, we note that $P_\gamma$ for each $\gamma\in\Gamma$ is a projective operator on ${\cal H}$. Then we assume that $$P_\gamma=\sum_j|\gamma_j\rangle \langle \gamma_j|$$  where $\{|\gamma_j\rangle\}$ is an orthonormal set.
As a result, $\Prob_{{\cal A},\gamma}(x)$ can be rewritten as
 \begin{eqnarray}  \Prob_{{\cal A},\gamma}(x) &=&\sum_j|\langle \gamma_j|M(x) |\phi_0\rangle|^ 2\\
    &=& \sum_{j}(\langle \gamma_j|\otimes \langle (\gamma_j|)^{*})(M(x)\otimes M(x)^{*})|\phi_0\rangle \otimes (|\phi_0\rangle)^{*}
\end{eqnarray}
where $*$ denotes the conjugate operation.

Therefore,  for 1QFAC ${\cal A}=(S,Q,\Sigma,\Gamma, s_{0},|\psi_{0}\rangle, \delta,
\mathbb{U}, {\cal M})$ and $\gamma\in\Gamma$, we obtain a BLM given by
\begin{equation}
BLM({\cal A},\gamma)=(S_{{\cal A}},|\phi_0\rangle \otimes (|\phi_0\rangle)^{*}, \{M(\sigma)\otimes M(\sigma)^{*}\}_{\sigma\in\Sigma},\sum_{j}(\langle \gamma_j|\otimes \langle (\gamma_j|)^{*})) \label{BLM}
\end{equation}
where $|S_{{\cal A}}|=kn$   with $|S|=k$ and $|Q|=n$, such that
\begin{equation}\Prob_{{\cal A},\gamma}(x)=f_{BLM({\cal A},\gamma)}(x)\label{Sim}\end{equation}
for $x\in\Sigma^*$.

In summary, for a 1QFAC ${\cal A}$ with the output alphabet $\Gamma$, we obtain a family of BLM $\{BLM({\cal A},\gamma): \gamma\in \Gamma\}$ which equivalently simulate  the behavior of ${\cal A}$ (i.e., Eq. (\ref{Sim}) holds for any $\gamma\in\Gamma$ and $x\in\Sigma^*$).  It  is worth stressing that all BLM in the family have the same structure except for the final vectors. Note that if ${\cal A}$ is a language acceptor, i.e., $\Gamma=\{a, r\}$, then only one $BLM({\cal A},a)$ is sufficient to simulate ${\cal A}$ since it  holds that $\Prob_{{\cal A},a}(x)+\Prob_{{\cal A},r}(x) =1$ for all $x\in\Sigma^*$. Indeed, if $|\Gamma|=m$, then $m-1$  BLM like the one given in (\ref{BLM})   are sufficient to simulate ${\cal A}$.

Based on the above discussion and  Proposition \ref{BLM}, we can obtain that two 1QFAC ${\cal A}_{1}=(S_{1},Q_{1},\Sigma,\Gamma,
s_{0},|\psi_{0}^{(1)}\rangle, \delta_{1}, \mathbb{U}_{1}, {\cal
M}_{1})$ and ${\cal A}_{2}=(S_{2},Q_{2},\Sigma,\Gamma,
t_{0},|\psi_{0}^{(2)}\rangle, \delta_{2}, \mathbb{U}_{2}, {\cal
M}_{2})$ are equivalent if and only if they are $(k_1 n_1)^{2}+(k_2 n_2)^{2}-1$-equivalent, where $k_i$ and $n_i$ are the numbers of classical and quantum basis states of ${\cal A}_{i}$, respectively, $i=1,2$. In addition,  there exists a polynomial-time algorithm running in
time $O([(k_1 n_1)^{2}+(k_2 n_2)^{2}]^4)$ that takes as input two 1QFAC ${\cal A}_1$
and ${\cal A}_2$ and determines whether ${\cal A}_1$ and ${\cal
A}_2$ are equivalent. We formulate this result as follows.

\begin{thm} \label{thm:bilbound}
Two 1QFAC ${\cal A}_1$ and
${\cal A}_2$  are
equivalent if and only if they are $(k_1 n_1)^{2}+(k_2 n_2)^{2}-1$-equivalent.
Furthermore, there exists a polynomial-time algorithm running in
time $O([(k_1 n_1)^{2}+(k_2 n_2)^{2}]^4)$ that takes as input two 1QFAC ${\cal A}_1$
and ${\cal A}_2$ and determines whether ${\cal A}_1$ and ${\cal
A}_2$ are equivalent, where  $k_i$ and $n_i$ are the numbers of classical and quantum basis states of ${\cal A}_{i}$, respectively, $i=1,2$.

\end{thm}


\section{Minimization of 1QFAC}\label{sec5}

In this section we show that the minimization of 1QFAC is decidable. The result relies on Theorem~\ref{thm:bilbound} and on the decidability of the theory of real ordered fields \cite{GV88, Ren88}.
It also requires that {\em we only use algebraic complex numbers when defining automata}.
This does not raise theoretical difficulties because all quantum states reachable by such an automaton remain in the linear space over the field of algebraic complex numbers. Furthermore, this assumption is not a practical restriction. Indeed, the set of algebraic complex numbers is dense. Moreover, it contains all floating-point numbers and all rational numbers. Finally, the usual set of universal quantum gates is defined only with algebraic complex numbers~\cite{boy:99}.

Indeed, the present method has already been used for the minimization of multi-letter 1QFA \cite{QLZ11}, MO-1QFA and MM-1QFA \cite{PQL12}. However, 1QFAC contain both classical and quantum states, and both states will be considered to be minimized simultaneously. In the interest of readability, we would describe the minimization process of 1QFAC in detail.

We start by briefly recalling the decision problem for the existential theory of the reals \cite{GV88}, that is, the problem of deciding if the set $\mathbb{S}=\{x\in \mathbb{R}^n: \mathbf{P}(x)\}$ is non-empty, where $\mathbf{P}(x)$ is a predicate which is a Boolean function of atomic predicates either of the form $f_i(x)\geq 0$ or $f_j(x)>0$, the $f'$s being real polynomials.  After  \cite{GV88}, many authors have studied this problem (for example, J. Canny \cite{Can88},  J. Heintz, and J. Renegar et al. \cite{Ren88}), and here we recall Renegar's result \cite{Ren88}. More precisely,  Renegar \cite{Ren88} designed an algorithm of time complexity $(nd)^{O(k)}$ solving  the problem of determining  if the set $\mathbb{S}$ defined above is non-empty, where $d$ is the degree of polynomials, $k$ the number of variables, and $n$ the number of polynomials.  Furthermore, to find
a sample of $\mathbb{S}$ requires $\tau d^{O(n)}$ space if all coefficients of the atomic predicates use at most $\tau$
space (see \cite{Ban03}, page 518), which means that finding a sample requires
exponential space on the number of variables. Let us summarize these results  in the following theorem.

\begin{thm}[\cite{Can88,Ren88,Ban03}]\label{BCR}
To decide whether the set $\mathbb{S}=\{x\in \mathbb{R}^n: \mathbf{P}(x)\}$ is non-empty, where $\mathbf{P}(x)$ is a predicate which is a Boolean function of atomic predicates either of the form $f_i(x)\geq 0$ or $f_j(x)>0$, with $f'$s being real polynomials (with integer coefficients), can be done in PSPACE in  $n, m, d$, where $n$ is the number of variables, $m$ is the number of atomic predicates, $d$ is the highest degree among all  atomic predicates of $\mathbf{P}(x)$.  Moreover, there exists an algorithm of time complexity $(md)^{O(n)}$ for this problem. To find a sample of  $\mathbb{S}$ requires  $\tau d^{O(n)}$ space if all coefficients of the
atomic predicates use at most $\tau$ space.
\end{thm}

Since 1QFAC are usually defined over the field of complex numbers,  we need to transform the problem over  the field of complex numbers to that over real numbers. That will be based on the following observation.

\begin{Rm}\label{Rm-BCR}
Any complex number $z=x+yi$ is determined by two reals $x$ and $y$, and any complex polynomial   $f(z)$ with $z\in \mathbb{C}^n$ can be equivalently written as $f(z)=f_{1}(x, y)+if_{2}(x, y)$ where $(x,y)\in \mathbb{R}^{2n}$ is the real representation of $z$, and $f_{1}$ and  $f_{2}$ are real polynomials. Thus, a system of  $n$ complex polynomial equations  with $k$ complex variables   can be equivalently described by a system of $2n$ real polynomial equations with $2k$ real variables.
\end{Rm}
Of course, regarding the problem of solving a system of polynomial equations, we can also refer to the work by A. Fr\"{u}hbis-Kr\"{u}ger and C. Lossen \cite{KL05,Los03}, and they studied this problem in detail.

Now we  deal with the minimization of 1QFAC.
Consider the set $\mathrm{N}=\{(k,n):k\geq 1, n\geq 1\}$ where $k,n$ are integer. Then the number of states of any 1QFAC belongs to $\mathrm{N}$ in which the first element denotes the number of classical states and the second one  the number of quantum basis states.

Assume we are given a well-ordered relation, say $\preceq$, over $\mathrm{N}$ where the smallest element is $(1,1)$. What follows does not depend on this choice. From a practical point of view, different choices can be justified depending on the goals of the user. For instance the user may wish to give priority to reducing the qubits needed to implement the automaton (that is, minimize $n$ even at the expense of using more classical states). In this case the user might want to use the well ordering induced by the following strict order:
$$(k,n)\prec(k',n')\textrm{ iff }k + 2^n < k' + 2^{n'} \textrm{ or } (k + 2^n = k' + 2^{n'} \text{ and } n < n').$$

Given a 1QFAC ${\cal A}=(S,Q,\Sigma,\Gamma, s_{0},|\psi_{0}\rangle, \delta,
\mathbb{U}, {\cal M})$, where we suppose that the numbers of classical states and quantum basis states are $k$ and $n$, respectively. Then, according to the well-ordered relation $\preceq$, we have $(1,1)\preceq (k_1,n_1)\preceq (k_2,n_2)\preceq\ldots\preceq (k,n)$ for all elements $(k_i,n_i)$ ``smaller" than $(k,n)$. The minimization of 1QFAC ${\cal A}$ is to search for the minimal pair  $(k_i,n_i)\preceq (k,n)$ for which there is a 1QFAC ${\cal A}_{\min}$ with $k_i$ classical states and $n_i$ quantum basis states equivalent to ${\cal A}$. To this end, we check it from $(1,1)$ to $(k_i,n_i)$ step by step.

First, we prove that, for any $(k',n')$, the problem of whether there is a 1QFAC ${\cal A}'$ with $k'$ classical states and $n'$ quantum basis states equivalent to ${\cal A}$ is decidable. Without loss of generality, for simplicity, we consider ${\cal A}$ to be an acceptor, i.e., $\Gamma=\{a,r\}$.

\begin{Lm} \label{decidable} Given an acceptor 1QFAC ${\cal A}$ and given a pair $(k',n')$, the problem of whether there exists an equivalent acceptor 1QFAC ${\cal A}'$ with $k'$ classical states and $n'$ quantum basis states is decidable in EXPTIME in $n'$ and $k'$.
\end{Lm}

\begin{proof} Suppose ${\cal A}=(S,Q,\Sigma,\Gamma, s_{0},|\psi_{0}\rangle, \delta,
\mathbb{U}, {\cal M})$ where $\Gamma=\{a,r\}$.  Let $S'$ denote a set of classical states and $Q'$ a set of quantum basis states, where $|S'|=k'$ and  $|Q'|=n'$. We know that the number of different mappings from $S'\times \Sigma$ to $S'$ is $(k')^{k'\times |\Sigma|}$. For any given mapping $\delta': S'\times \Sigma\rightarrow S'$, we check whether there is a 1QFAC ${\cal A}'$ with transition $\delta'$ equivalent to ${\cal A}$. We will prove this is decidable. If there is  a 1QFAC ${\cal A}'$ with transition $\delta'$ equivalent to ${\cal A}$, then we claim that the state complexity of ${\cal A}$ can be reduced to $(k',n')$. If  for any $\delta'$, there is no 1QFAC ${\cal A}'$ with transition $\delta'$ equivalent to ${\cal A}$, then the state complexity of ${\cal A}$ can not be reduced to $(k',n')$. Thus, the key is to prove that the problem of whether there is a 1QFAC ${\cal A}'$ with transition $\delta'$ equivalent to ${\cal A}$ is decidable.

Suppose that there exists such a ${\cal A}'$ with transition $\delta'$ equivalent to ${\cal A}$.
We let ${\cal A}'=(S',Q',\Sigma,\Gamma, s_{0}',|\psi_{0}'\rangle, \delta',
\mathbb{U}', {\cal M}')$ where, for each $s'\in S'$ and each $\sigma\in \Sigma$, $U_{s'\sigma}\in \mathbb{U}'$ and suppose
\[ U_{s'\sigma}=[x_{ij}(s'\sigma)]
\]
which is an $n'\times n'$ unitary matrix and therefore satisfies
\begin{equation}
[x_{ij}(s'\sigma)]\times [x_{ij}(s'\sigma)]^{\dag}=I \label{m1}
\end{equation}
where $\dag$ denotes the conjugate transpose operation.   Thus by Remark \ref{Rm-BCR} we can use $2{n'}^2$ real polynomial equations with $2{n'}^2$ real variables to describe that  $U_{s'\sigma}$ is a unitary matrix. Note that we should describe $U(\sigma)$ for every $\sigma\in\Sigma$ and every $s'\in Q'$.   Thus, the number of $U_{s'\sigma}$'s is $k'|\Sigma|$.

Suppose $|\psi_{0}'\rangle=[y_1,y_2,\cdots, y_{n'}]^{T}$ where $T$ denotes the transpose operation. Then
\begin{equation}
\sum_{i=1}^{n'}y_iy_i^*=1 \label{m2}
\end{equation}
where $*$ denotes the conjugate operation. Thus  we can use two real polynomial equations with $2n'$ real variables to describe that $|\psi_{0}'\rangle$ is a unit vector in $\mathbb{C}^{n'}$.

Regarding the  projection measurement set ${\cal M}'=\{{\cal M}_{s'}\}_{s'\in S'}$, there are also finite cases since $Q'$ is finite. More exactly, there are $2^{n'}$ cases for each $s'\in S'$.  Also, we need to check it for each case for each $s'\in S'$. Suppose that ${\cal M}_{s'}=\{P_{s',a}, I-P_{s',a}\}$. Then we can  describe $P_{s',a}$ as follows:
\begin{equation}P_{s',a}=\sum_{i=1}^{n'}z_{(s',i)}|q_i\rangle\langle q_i|\end{equation}
with
\begin{equation} z_{(s',i)}=1 \text{~or~} z_{(s',i)}=0,\label{P} \end{equation}
where  $z_{(s',i)}=1$ means that the state $q_i\in Q$ should be regarded as  an accepting state, otherwise a rejecting state. Therefore, for each $s'\in S'$ we can use $2n'$ real polynomial equations with $n'$ real variables to describe the projective measurement ${\cal M}_{s'}$.

Since ${\cal A}'$ is equivalent to ${\cal A}$, for each $x\in \Sigma^*$ with $|x|\leq(k n)^{2}+(k' n')^{2}-1$, by Theorem \ref{thm:bilbound} we have the following equations:
\begin{equation}
\|P_{\mu'(x),a}v'(x)|\psi_{0}'\rangle\|^2=\|P_{\mu(x),a}v(x)|\psi_{0}\rangle\|^2 \label{m3}
\end{equation}
where $\mu'(x)$ and $\mu(x)$ denote respectively the classical states of ${\cal A}'$ and ${\cal A}$  after reading $x$,  $v'(x)$ and $v(x)$ the unitary operators ${\cal A}'$ and ${\cal A}$  for reading $x$ defined as Eq. (\ref{v}), respectively. Note that in the above equation, the right side is   a fixed value for the given 1QFAC ${\cal A}$ (of course, some time is need to compute this value for the given ${\cal A}$), and  the left side can be rewritten as
\begin{equation}
\Prob_{{\cal A},a}(x)=\sum_{i=1}^{n'} z_{(\mu'(x),i)}\langle q_i|\otimes\langle q_i| v'(x)\otimes v '(x)^*|\psi_{0}'\rangle\otimes|\psi_{0}'\rangle^*
\end{equation}
where we assume that $P_{\mu'(x),a}=\sum_{i=1}^{n'}z_{(\mu'(x),i)}|q_i\rangle\langle q_i|.$ Thus, the left side of Eq. (\ref {m3}) can be described by a real polynomial, of which the degree is $2|x|+3$. Note that to describe the fact that ${\cal A}'$ and ${\cal A}$ are equivalent, the total number of polynomial equations needed is
\begin{equation}
P=|\Sigma|^1+|\Sigma|^2+\cdots+ |\Sigma|^{(k n)^{2}+(k' n')^{2}-1}.
\end{equation}

The above statements and analysis can now  be summarized as follows:
  for a given  acceptor    1QFAC ${\cal A}$ over an input alphabet $\Sigma$, another 1QFAC ${\cal A}'$ with a given classical transition function $\delta'$ that is equivalent to ${\cal A}$ can be represented by a vector $x\in \mathbb{R}^{ 2k'|\Sigma|n'^2+(2+k')n'}$,  which is restricted by  these real polynomial equations from Eqs. (\ref{m1},\ref{m2},\ref{P},\ref{m3}). The total number of the  polynomial equations needed is
 \begin{equation}
N=2+2k'|\Sigma|{n'}^2+2k'n'+P.
 \end{equation}
 The highest degree in these equations is
   \begin{equation}d=2[(k n)^{2}+(k' n')^{2}-1]+3.\end{equation} Thus, according to  Renegar's algorithm \cite{Ren88} as we reviewed above (Theorem \ref{BCR}), it is decidable that whether or not there exists a vector $x\in \mathbb{R}^{ 2k'|\Sigma|n'^2+(2+k')n'}$ satisfying these real polynomial equations from Eqs. (\ref{m1},\ref{m2},\ref{P},\ref{m3}), and its time complexity is
 \begin{equation}
 T=(Nd)^{O( k'|\Sigma|n'^2)}.
  \end{equation}
If it has a solution, then
${\cal A}'$ is equivalent to ${\cal A}$, from which it follows that the state complexity of ${\cal A}$ can be reduced to $(k',n')$.

In summary, if the above all cases have been checked and there is no solution for these equations (\ref{m1},\ref{m2},\ref{P},\ref{m3}), then we can conclude that the state complexity of ${\cal A}$ can not be reduced to $(k',n')$. Otherwise, the state complexity of ${\cal A}$ can  be reduced to $(k',n')$.
\end{proof}

Using Lemma \ref{decidable}, the envisaged result is immediate.

\begin{thm} Given an acceptor 1QFAC ${\cal A}$ with $k$ classical states and $n$ quantum basis states, the minimization problem of ${\cal A}$ is decidable in EXPSPACE on $k$ and $n$.
\end{thm}

\begin{proof} Suppose that ${\cal A}$  has an input alphabet $\Sigma$. For a pair $(k',n')$  chosen from $(1,1)$ to $(k,n)$, we construct a classical state set $S'$ such that $|S'|=k'$. Then as mentioned before the number of different mappings from $S'\times \Sigma$ to $S'$ is $(k')^{k'\times |\Sigma|}$, and we denote the set of all these mappings by $Map(S',\Sigma,S')$. For each $\delta'\in Map(S',\Sigma,S')$, we define the set
\begin{eqnarray}
\mathbb{S}_{{\cal A},\Sigma}^{(k',n',\delta')}&=&\{{\cal A}': {\cal A}' \text{~is a 1QFAC equivalent to~}{\cal A} \text{ over~} \Sigma \text{,~with~} \text{state number pair~} \nonumber \\ &~~&(k',n') \text{~ and with classical transition function~} \delta'\}\nonumber.
\end{eqnarray}
Thus by taking $(k',n')$ from $(1,1)$ to $(k,n)$, by Lemma \ref{decidable} we check whether or not 1QFAC ${\cal A}$ can be reduced to another 1QFAC ${\cal A}'$ with  $k'$ classical states and $n'$ quantum basis states.  The minimization algorithm is now depicted as follows:
\begin{center}

 Algorithm for the minimization of 1QFAC. \vskip 1mm
 \fbox{\parbox{\textwidth}{
 {\small\vskip 1mm
{\bf Input:} a 1QFAC ${\cal
A}$ with state number pair $(k,n)$\\
{\bf Output:} a minimal 1QFAC ${\cal
A}^{'}$ equivalent to ${\cal
A}$ with respect to a well-ordered relation $\preceq$\\
{\bf Step 1:}  \begin{quote} Take $(k',n')$ from $(1,1)$ to $(k,n)$.
\begin{quote}
Take $\delta'\in Map(S',\Sigma,S')$
\begin{quote} If ($\mathbb{S}_{{\cal A},\Sigma}^{(k',n',\delta')}$ is not empty) return ${\cal A}' =$ sample $\mathbb{S}_{{\cal A},\Sigma}^{(k',n',\delta')}$\end{quote}\end{quote}
\end{quote}
{\bf Step 2:}
 \begin{quote} return  ${\cal
A}'={\cal
A}$
\end{quote}

}}}
\end{center}

In the above algorithm,
the worst-case time complexity is $O(k'^{k'\times |\Sigma|}\times T)$ for checking whether a given $(k',n')$ has an automaton equivalent to ${\cal A}$.  If such an automaton exists, we can furthermore give a description on the automaton, i.e., to  find a sample of $\mathbb{S}_{{\cal A},\Sigma}^{(k',n',\delta')}$. According to Theorem \ref{BCR}, to find a sample needs exponential space. Anyway, we have presented an algorithm to find a minimal 1QFAC equivalent to a given 1QFAC.
\end{proof}

As we know, when a 1QFAC has only one classical state, it is exactly an MO-1QFA. Therefore, we obtain the minimization of MO-1QFA using the obvious well-ordering.

\begin{Co}
Given an MO-1QFA ${\cal A}$ with $n$ quantum basis states, the minimization problem of ${\cal A}$  is decidable.
\end{Co}

\begin{Rm}
The minimization problem of MO-1QFA  was proposed by Moore and Crutchfield (see \cite{MC00}, page 304, Problem 5) and we here present an answer to this problem. Since $k=1$, the worst-case  time  complexity is  $O\left({\left(n^4|\Sigma|+n^2|\Sigma|^{n^2}\right)}^{|\Sigma|n^2}\right)$.
\end{Rm}

\section{Conclusions and problems}

In this paper, we proposed a new model for one-way QFA, which we called 1QFAC. Such automata can accept all regular languages with no error, and, moreover, they can accept with bounded error some languages with essentially less states than DFA and for which there is no MO-1QFA, nor MM-1QFA, nor  multi-letter 1QFA accepting them. 1QFAC contain both classical and quantum components and  therefore, 1QFAC  inherit the characteristics of  DFA but improved on them by employing quantum computing. From the practical point of view,  1QFAC can be as physically realizable as MO-1QFA , and therefore it is, to a certain extent, a practical model of quantum computing with finite memory.


In detail, we addressed  the lower-bound state complexity problem of 1QFAC, and showed that, if $L$ is accepted by a 1QFAC ${\cal M}$ with bounded error, then $kn=\Omega( \log m)$ where $k$ and $n$ denote  numbers of classical states and quantum states of ${\cal M}$, respectively, and $m$ is the state number of the minimal DFA accepting $L$. We have proved this lower bound is tight (Proposition \ref{prop:accbe}). Indeed, we verified that, for any prime number $m\geq 2$, there exist some regular languages $L_m$ whose minimal DFA  needs $O(m)$ states, and there is no MO-1QFA, nor MM-1QFA nor  multi-letter 1QFA that  can accept $L_m$, but there exists 1QFAC accepting $L_m$ with only constant classical states (independent of $m$) and $\log m$ quantum basis states.
Also, we have proved that
any two 1QFAC ${\cal A}_1$ and ${\cal A}_2$ are equivalent if and
only if they are $(k_{1}n_1)^2+(k_{2}n_2)^{2}-1$-equivalent, where $k_{1}$ and
$k_{2}$ are the numbers of classical states of ${\cal A}_1$ and
${\cal A}_2$, as well as  $n_{1}$ and
$n_{2}$ are the numbers of quantum basis states of ${\cal A}_{1}$ and ${\cal
A}_{2}$, respectively; in addition, there exists a
polynomial-time $O((k_{1}n_1)^2+(k_{2}n_2)^{2})^{4})$ algorithm for determining their
equivalence. Finally, we have shown that minimization of 1QFAC is decidable in EXPSPACE. As a corollary of this result, we have also shown that the minimization problem of MO-1QFA is decidable, a problem proposed by Moore and Crutchfield (see \cite{MC00}, page 304, Problem 5).

To conclude, we would like to pose some open problems for further consideration.

\begin{itemize}

\item State complexity of 1QFAC: For any given regular language $L$, if the minimal number of states of the DFA accepting $L$ is $n$, then for any $n_{1}< n$, whether or not there exists a 1QFAC accepting $L$  with $n_{1}$ classical states and some quantum basis states?

 \item 1QFA {\it with control languages} (1QFACL), the ancilla 1QFA in \cite{Pas00}, and the Ciamarra 1QFA in
\cite{Cia01} also accept all regular languages \cite{LQZLW09}, and Remark \ref{1QFACL} shows a certain advantage of 1QFAC over 1QFACL in state complexity. Compare the state complexity of 1QFAC with these 1QFA in detail, and discover more languages to verify the advantage of 1QFAC over MO-1QFA or other 1QFA concerning the space-efficiency of states? (Here we would like to stress MO-1QFA because 1QFAC may be thought of as an generalization of MO-1QFA with inheriting the component of classical DFA.)

\item What would be the consequences of relaxing the notion of equivalence between automata to equivalence up to $\varepsilon$? More precisely, for instance, one should investigate the equivalence problem when two automata are considered equivalent iff their acceptance probability distributions over the strings do not differ more than $\varepsilon$ at each string.


\end{itemize}

\subsubsection*{Acknowledgments}
This work is supported in part by the National
Natural Science Foundation (Nos. 61272058, 61073054, 60873055, 61100001), the Natural
Science Foundation of Guangdong Province of China (No.
10251027501000004),   the Specialized Research Fund for the Doctoral Program of Higher Education of China
(Nos. 20100171110042, 20100171120051), the Fundamental Research Funds for the Central Universities (No. 11lgpy36), and the project
of  SQIG at IT, funded by FCT and EU FEDER projects QSec
PTDC/EIA/67661/2006, AMDSC UTAustin/MAT/0057/2008, NoE Euro-NF,
and IT Project QuantTel, FCT project PTDC/EEA-TEL/103402/2008
QuantPrivTel, FCT PEst-OE/EEI/LA0008/2013.

\end{document}